\documentclass[a4paper]{llncs}
\usepackage[utf8]{inputenc}
\usepackage{amsmath}
\usepackage{amssymb}
\usepackage{color}
\usepackage{tikz}
\usepackage{pgffor}
\usepackage{authblk}
\usepackage{hyperref}
\usepackage{float}
\usepackage{etoolbox}
\usepackage{wrapfig}

\usetikzlibrary{arrows,automata,shapes,positioning,calc,fit,shapes}

\newcommand{\keywords}[1]{\par\addvspace\baselineskip
        \noindent\keywordname\enspace\ignorespaces#1}

\newcommand{\set}[1]{\{ #1 \}}
\newcommand{\pset}[2]{\{\, #1 : #2 \,\}}

\newcommand{\struct}[1]{\mathcal{#1}}
\newcommand{\pol}{\mapsto}
\newcommand{\flows}{\leadsto}

\newcommand{\owofunc}[1]{\operatorname{#1}}
\newcommand{\obs}{\owofunc{obs}}
\newcommand{\dom}{\owofunc{dom}}
\newcommand{\view}{\owofunc{view}}
\newcommand{\Runs}{\owofunc{Runs}}
\newcommand{\purge}{\owofunc{purge}}
\newcommand{\act}{\owofunc{act}}

\newcommand{\last}{\owofunc{last}}

\newcommand{\Views}{\owofunc{Views}}
\newcommand{\aconc}{\,\hat{\circ}\,}
\newcommand{\mname}[1]{\text{#1}}

\newcommand{\absUp}[2]{\operatorname{Hu}(#1)}
\newcommand{\absUpPol}[2]{\operatorname{Hu}^{#2}(#1)}
\newcommand{\absDown}[2]{\operatorname{Ld}(#1)}
\newcommand{\absDownPol}[2]{\operatorname{Ld}^{#2}(#1)}

\newcommand{\proj}[2]{\operatorname{pr}_{#1}^{#2}}

\newcommand{\absn}{{\cal C}} 

\newtoggle{long}
\toggletrue{long}

\usepackage{version} 
\includeversion{shortv} 
\includeversion{longv} 

\mainmatter
\title{\normalfont On Reductions from Multi-Domain Noninterference to the Two-Level Case}
\author{
  Oliver Woizekowski
  \and
  Ron van der Meyden
}
\institute{
    Department of Computer Science, Kiel University \\
  \email{oliver.woizekowski@email.uni-kiel.de}
  \and
  School of Computer Science and Engineering, 
  UNSW Australia \\
  \email{meyden@cse.unsw.edu.au}
}

\begin{document}

\maketitle

\begin{abstract}
The literature on information flow security
with respect to transitive policies 
has been concentrated largely on the case of policies with two security domains, High and Low, 
because of a presumption that more general policies can be reduced to this two-domain case.
The details of the reduction have not been the subject of careful study, however. 
Many works in the literature use a reduction based on 
a quantification over ``Low-down''  partitionings of domains into those 
below and those not below a given domain in the information flow order. 
A few use  ``High-up" partitionings of domains into those above and those not above a given domain. 
Our paper argues that more general ``cut" partitionings are 
also appropriate, and studies the relationships between the resulting multi-domain notions of security
when the basic notion for the two-domain case to which we reduce is either Nondeducibility on Inputs or Generalized Noninterference.
The Low-down reduction is  shown to be weaker than the others, and while the High-up reduction is 
sometimes equivalent to the cut reduction, both it and the Low-down reduction may have 
an undesirable property of non-monotonicity with respect to a natural ordering on  policies. 
These results suggest that the cut-based partitioning yields a more robust general approach 
for reduction to the two-domain case.

  \keywords{Noninterference, nondeterminism, information flow, covert channels, policies}
\end{abstract}


\section{Introduction}

Information flow security is concerned with finding, preventing and understanding the unwanted flow of information within a system implementation.
One of its applications is the detection of covert channels,
which might arise due to hard-to-foresee side-effects in the combination of smaller components,
or even have been deliberately planted in the implementation by a rogue systems designer.

In order to reason about information flow, one needs to decompose the system into information \emph{domains}.  
Domains are thought of as active components (users, processes, pieces of hardware, organisational units, etc.) and change the system state by performing actions.
Domains may also make observations of the system state. One way for information to flow from one domain to another is for the actions of the first to change the 
observations of the second. 
To describe the allowed flows of information in the system, one can 
specify for each pair of domains in which directions a flow of information is permissible.
This specification is called a \emph{policy} and usually represented as a directed graph: two examples are depicted in Figure \ref{fig:policy-examples}.
Policies are generally taken to be reflexive relations, since nothing can prevent a domain from obtaining
information about itself. Moreover, they are often assumed to be transitive, 
(i.e., if $A \pol B$ and $B \pol C$ then we must also have $A \pol C$) 
since if $B$ may obtain information about $A$, and $B$ may pass this information to $C$, then 
there is nothing to prevent $C$ receiving information about $A$.
\footnote{We confine our attention in this paper to the transitive case. 
Works that have investigated intransitive information flow theory include \cite{HaighYoung87}, \cite{rushby_92} and \cite{RonEsorics2007}.} 

\begin{center}
  \vspace{-0.75cm}
  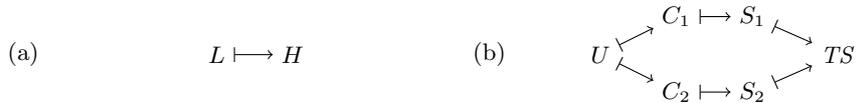
\begin{figure}[h]
  (a)
  \begin{minipage}{0.45\textwidth}
    \centering
    \begin{tikzpicture}[|->, auto]
      \node at (0,0) (L) {$L$};
      \node at (1,0) (H) {$H$};
      \path (L) edge node {} (H);
    \end{tikzpicture}
  \end{minipage}
  (b)
  \begin{minipage}{0.45\textwidth}
    \centering
    \begin{tikzpicture}[|->, auto]
      \node at (0, 0) (U) {$U$};
      \node at (1, +0.5) (C1) {$C_1$};
      \node at (1, -0.5) (C2) {$C_2$};
      \node at (2, +0.5) (S1) {$S_1$};
      \node at (2, -0.5) (S2) {$S_2$};
      \node at (3.125, 0) (TS) {$\mathit{TS}$};

      \path (U) edge node {} (C1);
      \path (U) edge node {} (C2);
      \path (C1) edge node {} (S1);
      \path (C2) edge node {} (S2);
      \path (S1) edge node {} (TS);
      \path (S2) edge node {} (TS);
    \end{tikzpicture}
  \end{minipage}
  \caption{The two-level policy $H \not\pol L$ and a transitive MLS-style policy.}
  \label{fig:policy-examples}
  \end{figure}
  \vspace{-0.75cm}
\end{center}

Policy (a) in Figure~\ref{fig:policy-examples}, which we call $H \not\pol L$,  
is  the simplest and most-studied case. 
Here  we have two domains $H$ and $L$, where $H$ is thought to possess high and $L$ low level clearance in the system, and 
information flow is permitted from $L$ to $H$, but prohibited in the other direction. 
In practice, a larger set of domains is used to represent different security classifications,
such as Unclassified ($U$), Confidential ($C$), Secret ($S$) and Top Secret $(TS$), and each security level 
may moreover be partitioned into compartments representing different types of information 
relevant to `need to know' restrictions. This leads to policies such as the transitive policy 
whose Hasse diagram is depicted in Figure~\ref{fig:policy-examples}(b). Here the 
Confidential classification has two independent compartment  domains ($C_1,C_2$), as does the 
Secret classification  ($S_1, S_2$).  

Informally, the statement $u \pol v$ can be read as ``$u$'s behaviour may influence $v$'s observations'' or ``$v$ may deduce something about $u$'s behaviour''.
A first formal definition for this intuition, called \emph{noninterference} was given by Goguen and Meseguer \cite{GM82}, 
in the context of a \emph{deterministic} automaton-based model.
A generalization to nondeterministic systems is desirable
so one can extend information flow analysis to, for example, the use of unreliable components, randomness or underspecification.
Several works (e.g., \cite{Sutherland86,McCullough88Ulysses,McLean94,Mantel2000Possibilistic,focardi_01,roscoe_95,Ryan2000}) 
extended the theory to nondeterministic systems
and richer semantic models such as process algebras,
resulting
in a multitude of security definitions for several kinds of models,
and with different intentions in mind.

Much of this subsequent literature has confined itself to the 
two-domain policy $H \not\pol L$,  because there has been a view  that more complex
policies can be treated by reduction to this case. 
One obvious way to do so, that we may call the \emph{pointwise} approach,
is to apply a two-domain notion of noninterference for each pair of 
domains $u,v$ in the policy with $u \not \pol v$. However, even in the case 
of deterministic systems, this can be shown to fail to detect situations where 
a domain may have disjunctive knowledge about a pair of other domains, 
neither of which may interfere with it individually (we present an example
of this in Section~\ref{sec:abstraction}).  
Goguen and Meseguer  \cite{GM82} already address this deficiency 
by what we may call a \emph{setwise} approach, which requires
that for each domain $u$, the set of domains $v$ with $v\not \pol u$
does not collectively interfere with $u$. 

However, while the setwise definition deals with what an individual
domain may learn about a group of other domains, it does not
deal with what groups may learn about individuals, or other groups. 
Subsequent work in the literature has taken this issue
of \emph{collusion} into account in reducing to the two-domain case. 
For example, a survey by 
Ryan \cite{Ryan2000} states:
\begin{quote}
It might seem that we have lost generality by assuming that the alphabet of the system is partitioned into High and Low. In fact we can deal with more general MLS-style policy with a lattice of classifications by a set of non-interference constraints corresponding to the various lattice points. For each lattice point $l$ we define High to be the union of the interfaces of agents whose clearance dominates that of $l$. Low will be the complement, i.e., the union of the interfaces of all agents whose clearance does not dominate that of $l$. Notice also that we are assuming that we can clump all the high-level users together and similarly all the low-level users. There is nothing to stop all the low users from colluding. Similarly any high-level user potentially has access to the inputs of all other high users. We are thus again making a worst-case assumption.
\end{quote}
We call the kind of groupings that Ryan describes \emph{High-up coalitions},
and interpret his comments as the suggestion to extend existing, already understood security definitions 
for $H \not\pol L$ to the multi-domain case by generating multiple instances of $H \not\pol L$ formed from the policy in question using High-up coalitions.
Ryan's High-up approach is used in some works (e.g., \cite{forster_97}), 
but many others (e.g., \cite{Mantel_thesis,millen_94,RoscoeWW96,sutherland_86}) 
use instead a dual notion of \emph{Low-down coalitions}, 
where for some domain $l$, the group $L$  is taken to be the set of domains $u$ with $u\pol l$ and $H$ is taken to be the  complement of this set. 

Yet other groupings exist that are neither High-up nor Low-down coalitions. For example,  in 
Figure~\ref{fig:policy-examples}(b), the grouping $L = \{U,C_1,C_2\}$ and $H = \{S_1,S_2,TS\}$, corresponds
to neither a High-up nor a Low-down coalition. It seems no less reasonable to consider $L$ to be a colluding
group that is seeking to obtain $H$ level information. Note that this grouping is a 
\emph{cut} in the sense that there is no $u \in H$ and $v \in L$ such that $u \pol v$. 
Since in such a cut, domains in $L$ cannot individually obtain information about domains in $H$, 
it is reasonable to expect that they should not be able to get such information collectively. 
This motivates a reduction to the two-domain case that quantifies over all cuts. 

Our contribution in this paper is to consider this range of alternative reductions from multi-domain policies to the 
two-domain case, and to develop an understanding of how these definitions are related and 
which are reasonable. Reductions must start with an existing notion of security for the two-domain case. 
We work with two basic security definitions: Generalized Noninterference, which was introduced in \cite{McCullough88},
and Nondeducibility on Inputs, first presented in \cite{Sutherland86}.
Our analysis shows that the relationships between the resulting notions of security are subtle, and
the adequacy of 
a reduction approach may 
depend on the base notion for the two-domain policy. 
Amongst other results, we show that: 
\begin{enumerate}
\item 
When the basic notion for the two-domain case is Generalized Noninterference, 
High-up coalitions yield a notion that is strictly stronger than the 
notion based on Low-down coalitions, which in turn is stronger than the 
pointwise generalization.
For Nondeducibility on Inputs, however, 
 High-up coalitions and Low-down coalitions give independent notions of security. 
Low-down coalitions imply the setwise definition in this case, but High-up coalitions imply only the weaker
pointwise version. 
\item 
  For Generalized Noninterference, 
  High-up coalitions are 
   `complete' 
  in the sense of being equivalent to a reduction quantifying over all cuts.  
  However, this completeness result does not  hold for Nondeducibility on Inputs, 
  where cuts yield a stronger notion of security. 
\item Not all the resulting notions of security have an expected property of monotonicity with respect to a natural restrictiveness
  order on policies. 
  (Security of a system should be preserved when one relaxes policy constraints.) 
   In particular, High-up coalitions with respect to Nondeducibility on Inputs does not have this property, 
   and Low-down coalitions do not have this property for either Generalized Noninterference or Nondeducibility on Inputs. 
\end{enumerate}
These conclusions indicate that while Ryan's proposal to use High-up coalitions is sometimes adequate, a reduction that quantifies over the 
larger set of all cut coalitions seems to 
yield the most generally robust approach 
for reducing multi-domain policies to the two-domain case. 

The structure of the paper is as follows.
In Section~\ref{sec:model}, we introduce our model and show how systems and policies are described.
Our reductions will use two basic security definitions for two-domain policies 
that are recalled and generalized to their obvious pointwise 
versions for the multi-domain case in Section~\ref{sec:secdefs}. 
Section~\ref{sec:abstraction} gives some examples showing why 
the pointwise versions are still weaker than required, and 
it is necessary to consider reductions using groupings of domains. 
The range of reductions we consider are defined in Section~\ref{sec:reduction}.
Our main results are stated in Section~\ref{sec:main},
\iftoggle{long}{
  full proofs of which are given in Section~\ref{sec:proofs}.
}{
  and an outline of the proof technique involved to prove these results is
  presented in Section~\ref{sec:proofs}.
}
Finally, we conclude and motivate further research in Section~\ref{sec:concl}.


\section{Background: Systems and Policy Model}
\label{sec:model}

\paragraph{Notational 
conventions.}
Sequences are represented as $xyz$, or $x \cdot y \cdot z$ if it helps readability.
The set of finite sequences over a set $A$ is denoted $A^*$, 
and 
the empty sequence is denoted $\varepsilon$.
\iftoggle{long}{ 
  The length of $\alpha$ is written as $|\alpha|$.
}{
}
We write $\alpha(i)$ to denote the element with index $i$ of a sequence $\alpha$, where $i \in \mathbb{N}$, and the first element of $\alpha$ is $\alpha(0)$.
We let $\last(\alpha)$ be the last element of $\alpha$ if $\alpha$ is non-empty, and let it be undefined if $\alpha$ is empty.
If $X \subseteq A$  and $\alpha \in A^*$ then let $\alpha|_X$ be the subsequence of $\alpha$ with only elements from $X$ retained.
The set of total functions from $A$ to $B$ is denoted~$B^A$.

\paragraph{Systems.}
We use an automaton-based model similar to the original Goguen-Meseguer one from \cite{GM82}.
A \emph{system} is a structure $(S,A,O,D,\Delta,\obs,\dom,s_I)$ with $S$ a set of \emph{states},
$A$ a finite set of \emph{actions},
$D$ a finite set of domains with at least two members,
$O$ a finite set of \emph{observations} such that $A$ and $O$ are disjoint,
${\Delta \subseteq S {\times} A {\times} S}$ a (nondeterministic) transition relation,
$\obs \colon D {\times} S \to O$ an observation function,
$\dom \colon A \to D$ an assignment of actions to domains,
and $s_I$ the initial state.
We write $\obs_u(s)$ for $\obs(u,s)$.
The value $\obs_u(s)$ represents the observation the domain $u$ makes when the system 
is in 
state $s$.
Observations can also be interpreted as outputs from the system.
For an action $a$, the domain $\dom(a)$ is the domain from which $a$ originates.
The relation $\Delta$ is called \emph{deterministic} if for all $s, s', s'' \mathrel\in S$, $a \in A$:
if $(s,a,s') \in \Delta$ and $(s,a,s'') \in \Delta$ then $s'\mathrel=s''$.
We assume systems to be \emph{input-enabled}, i.e. that for every $s \in S$ and $a \in A$ there is $s' \in S$ with $(s,a,s') \in \Delta$.
The assumption of input-enabledness is made to guarantee that the domains' reasoning is based on their actions and observations only and cannot use system blocking behaviour as a source of information.

A \emph{run} of a system is a sequence $s_0 a_1 s_1 \ldots a_n s_n \in S(AS)^*$ such that for $i<n$,
we have ${(s_i,a_i,s_{i+1}) \in \Delta}$.
It is $\emph{initial}$ if $s_0 = s_I$.
If not explicitly mentioned otherwise, we always assume initial runs.
The set of initial runs of a system $\struct{M}$ will be denoted $\Runs(\struct{M})$.
For a run $r$, the subsequence of actions of $r$ is denoted $\act(r)$ and the subsequence of actions performed by a domain $u$ is denoted $\act_u(r)$.

\paragraph{Notational and diagrammatic conventions for systems.}
If $u$ is a domain and $A$ the action set of a system, we write $A_u$ for the set of actions $a$ with $\dom(a)=u$.
Similarly, for $X$ a set of domains we write $A_X$ for the set of actions $a$ with $\dom(a) \in X$.
Systems are depicted as directed graphs, where the vertices contain the state names.
Domain observations are written near the vertices that represent the states.
Edges are labelled with action names and represent transitions from one state to another.
The initial state is marked with an arrow that points to it.
Self-looping edges are omitted when possible to reduce clutter: thus, the lack of an edge labelled by action $a$
from state $s$ (as would be required by input-enabledness) implies the existence of edge $(s,a,s)$.

\paragraph{Modelling information by views.}
We will be interested in an asynchronous semantics for information, and capture asynchrony by 
treating sequences that differ only by stuttering observations as indistinguishable. 
This can also be described as no domain having access to a global clock.
Intuitively, systems can be imagined as distributed and domains as representing network hosts.
From this intuition it follows, for a given domain $u$, that local state changes within domains distinct from $u$
that do not provide a new observation to $u$
must not generate a copy of $u$'s current observation.
To this end,  we use an `absorptive concatenation'  operator $\aconc$ on sequences.
For all sequences $\alpha$ and $b_0\ldots b_n$ we let $\alpha \aconc \varepsilon = \alpha$ and
\begin{center}
  $\alpha \aconc b_0 \ldots b_n = \begin{cases}
    \alpha \aconc b_1 \ldots b_n & \mbox{if } \alpha \neq \varepsilon \mbox{ and } \last(\alpha)=b_0 \\
      (\alpha \cdot b_0) \aconc b_1 \ldots b_n    & \mbox{otherwise.}                           \\
  \end{cases}$
\end{center}
One can imagine $\alpha \aconc \beta$ as $\alpha \cdot \beta$ with stuttering  at the point of connection removed.
The information a domain acquires over the course of a run is modelled by the notion of \emph{view}.
Considering systems as networks suggests that, during a run, a domain can only directly see the actions performed by itself.
This is reflected in our definition of view by eliminating actions performed by all other domains.
For a domain $u$ the operator $\view_u \colon \Runs(\struct{M}) \to (A \cup O)^*$ is defined inductively:
for the base case $r= s_I$ let $\view_u(r) = \obs_u(s_I)$.
For all $r \in \Runs(\struct{M})$ of the form $r = r'as$, where $r'\in \Runs(\struct{M})$, $a\in A$ and $s\in S$, let 
\begin{center}
  $\view_u(r) = \begin{cases}
    \view_u(r') \cdot a \cdot \obs_u(s)  & \mbox{if } \dom(a)=u \\
    \view_u(r') \aconc \obs_u(s) & \mbox{otherwise.}
  \end{cases}$
\end{center}
An element $\view_u(r)$ is called a \emph{$u$ view}. 
The set of all $u$ views in system $\struct{M}$ is denoted $\Views_u(\struct{M})$.

For an example of a view, see the system in Figure \ref{fig:system-view-example}
(recall that we elide self-loops) and consider
the run ${r = s_I a s_1 b s_2 b s_2 a s_3}$;
the domains are given by the set $\set{A,B}$,
the domain assignment is given by $\dom(a)=A$ and $\dom(b)=B$,
and the observations made by domain $B$ are depicted near the state names.
We have $\view_B(r) = \bot b 1 b 1 2$.

\begin{wrapfigure}{r}[4pt]{0.4\textwidth}
  \vspace{-12pt}
    \centering
    \begin{tikzpicture}[->, auto, scale=1.25]

      \node at (-1, 0) (nirvana) {};
      \node at (0,0) (init) {$s_I$};
      \node at (0,1) (initPrime) {$s_1$};
      \node at (1,0) (s0) {$s_0$};
      \node at (1,1) (s1) {$s_2$};
      \node at (2,1) (s2) {$s_3$};

      \node at (0, -0.2) (initObs) {\tiny $\bot$};
      \node at (1, -0.2) (s0Obs) {\tiny $0$};
      \node at (0, 1+0.2) (initPrimeObs) {\tiny $\bot$};
      \node at (1, 1+0.2) (s1Obs) {\tiny $1$};
      \node at (2, 1+0.2) (s2Obs) {\tiny $2$};

      \path (nirvana) edge node {} (init)
      (init) edge node {$b$} (s0)
      (init) edge node {$a$} (initPrime)
      (initPrime) edge node {$b$} (s1)
      (s1) edge node {$a$} (s2)
      ;

    \end{tikzpicture}
    \caption{System example.}
  \label{fig:system-view-example}
  \vspace{-12pt}
\end{wrapfigure}
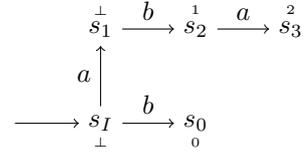
Note that $B$ does not notice the first transition in $r$ because we have $\obs_B(s_I) = \obs_B(s_1)$.
Domain $B$ does, however, learn about the last transition in $r$ due to $\obs_B(s_2) \neq \obs_B(s_3)$.
With the network analogy mentioned above, the last transition might model a communication from $A$ to $B$.

\paragraph{Policies.}
A \emph{policy} is a reflexive binary relation $\pol$ over a set of domains $D$.
We require $\pol$ to be reflexive because we assume that domains are aware of their own behaviour at all times.
We assume also that policies are transitive, to avoid additional complexities associated with the semantics of intransitive policies. 
Transitive policies arise naturally from lattices of security levels.
The policy that has received the most attention in the literature is over the set  $D=\set{H,L}$, consisting of a domain $H$ (or \emph{High}), representing
a high security domain whose activity needs to be protected, and a domain $L$ (or \emph{Low}), representing a low security attacker who aims to  learn High secrets. 
We refer to this policy as  $H \not\pol L$; it is given by the relation $\pol = \set{(H,H),(L,L),(L,H)}$.

If $\pol$ is a policy over some domain set $D$,
we write $u^{\pol}$ for the set $\pset{ v \in D }{ u \pol v }$,
and ${}^{\pol}u$ for the set $\pset{ v \in D  }{ v \pol u }$.
Similarly, the expression ${}^{\not\pol}u$ shall denote the set $\pset{ v \in D }{ v \not\pol u }$.

\paragraph{Further notational conventions for policies.}
Policies are depicted as directed graphs and their vertices carry domain names.
Edges due to reflexivity or transitivity are omitted.


\paragraph{Policy abstractions and cuts.}

A set of domains can be abstracted by grouping its elements into sets. 
Such groupings can be motivated in a number of ways. One is simply that 
we wish to take a coarser view of the system, and reduce the number of domains by 
treating several domains as one. Groupings may also arise from several domains 
deciding to collude in an attack on the security of the system.
Abstractions of a set of domains lead to associated abstractions of policies and systems.

An \emph{abstraction} of a set of domains $D$ is a set ${\cal D}$ of subsets of $D$
with $D = \bigcup_{F \in {\cal D}} F$ and $F \cap G \neq \emptyset$ implies $F=G$ for all $F,G \in {\cal D}$.
Associated with each abstraction ${\cal D}$ of $D$ is a function $f_{\cal D} \colon D\rightarrow {\cal D}$
defined by taking $f_{\cal D}(u)$ to be the unique $F \in {\cal D}$ with $u\in F$. 
For a policy $\pol$ over $D$ we let $\pol^{\cal D}$ be the policy over ${\cal D}$ defined by $F \pol^{\cal D} G$
if and only if there are $x \in F$ and $x' \in G$ with $x \pol x'$.


In order to formalize the idea of a reduction to $H \not\pol L$, we use abstractions that group all domains into
two sets that correspond to the High and Low domains.
A \emph{cut} of a set of domains $D$ with respect to a policy $\pol$
is a tuple $\absn = ({\cal H}, {\cal L})$
such that $\set{{\cal H}, {\cal L}}$ is an abstraction of $D$ and
there does not exist $u\in {\cal H}$ and $v\in {\cal L}$ with $u \pol v$.
When forming policies, we identify cuts with their underlying abstractions, and write $\pol^\absn$ for $\pol^{\set{{\cal H},{\cal L}}}$,
so the last requirement can also be formulated as ${\cal H} \not\pol^\absn {\cal L}$.
We mainly deal with abstractions that are given by cuts in this paper.
See Figure \ref{fig:cut-example} for an illustration of how policy (b) in Figure \ref{fig:policy-examples}
is abstracted using
${\cal C} := ({\cal H}, {\cal L}) = (\set{S_1, S_2, \mathit{TS}}, \set{U, C_1, C_2})$,
where we get ${\cal L} \pol^{\cal C} {\cal H}$ due to $C_1 \pol S_1$ or $C_2 \pol S_2$
and ${\cal H} \not\pol^{\cal C} {\cal L}$ as required for a cut.

\begin{center}
  \vspace{-0.75cm}
  \begin{figure}
   \centering
    \begin{minipage}{0.45\textwidth}
      \centering
      \begin{tikzpicture}[|->, auto]
        \node at (0, 0) (U) {$U$};
        \node at (1, +0.5) (C1) {$C_1$};
        \node at (1, -0.5) (C2) {$C_2$};
        \node at (2, +0.5) (S1) {$S_1$};
        \node at (2, -0.5) (S2) {$S_2$};
        \node at (3.125, 0) (TS) {$\mathit{TS}$};
        

        \path (U) edge node {} (C1);
        \path (U) edge node {} (C2);
        \path (C1) edge node {} (S1);
        \path (C2) edge node {} (S2);
        \path (S1) edge node {} (TS);
        \path (S2) edge node {} (TS);

        \draw [dotted] plot [smooth cycle] coordinates {
          ($(U) + (-0.2, 0)$)
          ($(C1) + (0.1, 0.2)$)
          ($(C2) + (0.1, -0.2)$)
        };
        \draw [dotted] plot [smooth cycle] coordinates {
          ($(TS) + (0.25, 0)$)
          ($(S1) + (-0.1, 0.2)$)
          ($(S2) + (-0.1, -0.2)$)
        };
        \node at (0.25, -0.6) (Low) {\tiny $\cal L$};
        \node at (2.8, 0.535) (High) {\tiny $\cal H$};
      \end{tikzpicture}
    \end{minipage}
    \caption{Illustration of a policy abstraction.}
    \label{fig:cut-example}
  \end{figure}
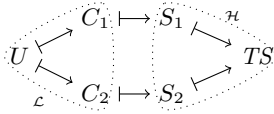
  \vspace{-0.75cm}
\end{center}

\paragraph{Systems and abstractions.}
Systems can be viewed from the perspective of an abstraction.
Intuitively, the actions of an abstract domain $F$ are all the actions of any 
of its subdomains $u\in F$.
It observes the collection of all observations made by the members of $F$
and thus their observations are functions from $F$ to $O$.
Let ${\struct{M}=(S,A,O,D,\Delta,\obs,\dom,s_I)}$ be a system and ${\cal D}$ be an abstraction of $D$.
Then $\struct{M}^{\cal D}$ is the system $(S,A,O',{\cal D},\Delta,\obs^{\cal D},\dom^{\cal D},s_I)$, where
$O'$ is the union of $O^F$ for all $F \in {\cal D}$,
its set of domains is ${\cal D}$,
for a state $s\in S$, the observation 
$\obs_F^{\cal D}(s)$ is the function with domain 
$F \in {\cal D}$ 
that sends each $x \in F$ to $\obs_x(s)$, and
$\dom^{\cal D}(a) 
= f_{\cal D}(\dom(a))$ for all $a \in A$. 
Intuitively, $\obs_F^{\cal D}(s)$ records the observations made in each domain in $F$ at $s$. 
Again, if $\absn = ({\cal H}, {\cal L})$ is a cut we write $\struct{M}^\absn$ for $\struct{M}^{\set{{\cal H},{\cal L}}}$.

\paragraph{Monotonicity with respect to restrictiveness.}
In \cite{RonSebastian2016} the notion of \emph{monotonicity with respect to restrictiveness} is discussed,
which holds for a given notion of security $X$ if, for all systems $\struct{M}$ and policies $\pol$ over the domain set of $\struct{M}$, the following statement holds:
if $\struct{M}$ is $X$-secure with respect to $\pol$
then $\struct{M}$ is $X$-secure with respect to every policy $\pol'$ with $\pol \subseteq \pol'$. 
If a notion of security satisfies this property, we will say that it is \emph{monotonic}.
Intuitively, adding edges to a policy reduces the set of information flow restrictions $u \not \pol v$ implied by the policy, 
making the policy easier to satisfy, so one would expect every sensible notion of security to be monotonic.
However, we will show that some notions of security obtained by a sensible construction based on cuts
do not support this intuition.

\section{Basic Notions of Noninterference}
\label{sec:secdefs}
In this section we recall two security definitions which have been proposed in the literature
for nondeterministic, asynchronous automaton-based models.
We use these as the basic definitions  of security for $H\not \pol L$ in the reductions  that we study. 
For purposes of comparison, we state the definitions using the most obvious pointwise generalization 
from the usual two-domain case to the general multi-domain case. 

\begin{longv} 
For deterministic systems, we define an operator $\cdot \colon S {\times} A \to S$ by $s \cdot a = s'$, where $s'$ is the unique state such that $(s,a,s') \in \Delta$ is satisfied.
This operator is extended to action sequences by setting $s \cdot \varepsilon = s$ and $s \cdot (\alpha a) = (s \cdot \alpha) \cdot a$.
This way action sequences `act' on the system's state set and $s \cdot \alpha$ is the state reached by performing $\alpha$ from $s$.

First, we recall Goguen and Meseguer's original notion of noninterference from \cite{GM84}.
They introduced a function $\purge$, which, for a given domain $u$, eliminates all actions that are supposed to be non-interfering with $u$.
It can be inductively defined as follows:
let $\purge_u(\varepsilon) = \varepsilon$ and for all $\alpha \in A^*$ and $a \in A$, let
\begin{center}
  $\purge_u(\alpha a) =
  \begin{cases}
    \purge_u(\alpha) \cdot a & \mbox{if } \dom(a) \pol u \\
    \purge_u(\alpha)         & \mbox{otherwise.}
  \end{cases}$
\end{center}
Their idea is, with respect to a fixed domain $u$, to deem a system secure if all action sequences that are equivalent under $\purge_u$ yield the same observations for $u$.

\begin{definition}
  \label{def:classical-security}
We say that $\struct{M}$ is \emph{$\mname{P}$-secure} for $\pol$ if for all $u \in D$ and $\alpha,\beta \in A^*$ we have:
if $\purge_u(\alpha)=\purge_u(\beta)$ then $\obs_u(s_I \cdot \alpha) = \obs_u(s_I \cdot \beta)$.
\end{definition}
This statement can be understood as the requirement that observations made by a domain may only depend on the behaviour by domains permitted to interfere with it.
\end{longv} 

\subsection{Nondeducibility on Inputs}

\begin{shortv} 
Goguen and Meseguer's definition of noninterference  \cite{GM84} was for deterministic systems only. 
\end{shortv} 
Historically, Sutherland \cite{Sutherland86} was the first to consider information flow in nondeterministic systems.
He presented a general scheme to instantiate notions of \emph{Nondeducibility}, i.e., epistemic definitions of absence of information flows. 
The notion of Nondeducibility on Inputs is one instance of this general scheme. 

Let $u,v \in D$.
We say that $\alpha \in {A_u}^*$ and $\beta \in \Views_v(\struct{M})$ are \emph{$v$ compatible} if there is $r \in \Runs(\struct{M})$ with $\act_u(r)=\alpha$ and $\view_v(r)=\beta$.
We write $u \flows_I v$ if there are $\alpha \in {A_u}^*$ and $\beta \in \Views_v(\struct{M})$ which are not $v$ compatible.
In that case $v$ gains information about $u$'s behaviour in the following sense:
if $\beta$ is observed by $v$ then $v$ can deduce that $u$ did not perform $\alpha$.
Nondeducibility $u \not \flows_I v$ therefore says that $v$ is unable to make any nontrivial deductions about $u$ behaviour. 
Applying this idea pointwise, we get the following definition of security: 

\begin{definition}
  A system is \emph{$\mname{NDI}_{pw}$-secure} for a policy $\pol$ over domains $D$ when 
  for all $u,v \in D$: if $u \not\pol v$ then $u \not\flows_I v$.
\end{definition}

In the case of the policy $H \not \pol L$ with just two domains,  
$\mname{NDI}_{pw}$ is the notion \emph{Nondeducibility on Inputs} as it is 
usually defined. We denote it as just $\mname{NDI}$ in this case. 
The definition above generalizes this notion in one possible
way to the multi-domain case. We discuss several others below.

\subsection{Generalized Noninterference}

The nondeducibility relation $H \not \flows L$ states that $L$ considers all sequences of actions of $H$ possible, 
but allows that $L$ has some information about how these actions, if any, are interleaved with 
$L$'s 
actions.
See Figure \ref{fig:NDI-weaker-than-GN} for a system that is $\mname{NDI}$-secure but can be argued to leak information about how $H$'s actions are
interleaved into a run.
The observations made by $L$ are written near the state names.
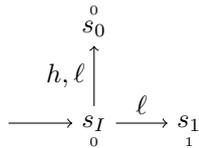
\begin{wrapfigure}{l}{0.5\textwidth}
  \vspace{-0.5cm}
    \centering
    \begin{tikzpicture}[->, auto, scale=1.25]

      \node at (-1, 0) (nirvana) {};
      \node at (0,0) (init) {$s_I$};
      \node at (1,0) (s1) {$s_1$};
      \node at (0,1) (s0) {$s_0$};

      \node at (0, -0.2) (initObs) {\tiny $0$};
      \node at (1, -0.2) (s1Obs) {\tiny $1$};
      \node at (0, 1+0.2) (s0Obs) {\tiny $0$};

      \path (nirvana) edge node {} (init)
      (init) edge node {$h,\ell$} (s0)
      (init) edge node {$\ell$} (s1)
      ;

    \end{tikzpicture}
  \caption{System demonstrating a weakness of  NDI.}
  \label{fig:NDI-weaker-than-GN}
  \vspace{-0.5cm}
\end{wrapfigure}
This system is $\mname{NDI}$-secure because every $L$ view is compatible with every possible sequence of $h$ actions performed by $H$.
However, note that if the view $0 \ell 1$ is observed by $L$ then it obtains the knowledge that it was the very first domain to act.
The stronger notion of \emph{Generalized Noninterference} introduced by McCullough \cite{McCullough88} 
says that $L$ does not have even this weaker form of knowledge. 
The original formulation is 
for a two-domain policy and is 
based on a model that uses sets of event sequences.
We present a straightforward multi-domain 
variant (that is similar to Mantel's combination BSI+BSD \cite{Mantel2000Possibilistic}). 


\begin{definition}
\label{def:GN}
A system $\struct{M}$ is \emph{$\mname{GN}_{pw}$-secure} for $\pol$ if for all
$u,v\in D$ with $u \not \pol v$, the properties
\begin{itemize}
\item $\mname{GN}^+(u,v)$: for all $r \in \Runs(\struct{M})$, 
for all $\alpha_0, \alpha_1 \in A^*$ 
with $\act(r)=\alpha_0\alpha_1$, and all $a \in A_u$ with 
there is $r' \in \Runs(\struct{M})$ with $\act(r')=\alpha_0 a \alpha_1$ and 
$\view_v(r')=\view_v(r)$, and

\item $\mname{GN}^-(u,v)$:
 for all $r \in \Runs(\struct{M})$, 
 all $\alpha_0,\alpha_1\in A^*$ and all $a\in A_u$, 
 with $\act(r)=\alpha_0 a \alpha_1$, 
 there is $r' \in \Runs(\struct{M})$ with  $\act(r')=\alpha_0\alpha_1$ and 
 $\view_v(r')=\view_v(r)$.
\end{itemize}
are satisfied.
\end{definition}

Intuitively, this definition says that actions of domains $u$ with $u \not \pol v$ can be arbitrarily inserted and deleted, 
without changing the set of possible views that $v$ can obtain. In the case of the 
two-domain policy $H \not \pol L$, the notion $\mname{GN}_{pw}$ is equivalent to 
the definition 
of Generalized Noninterference given in \cite{RonChenyi2007}, 
and we denote this case by $\mname{GN}$.
Note that the system in Figure \ref{fig:NDI-weaker-than-GN} is not $\mname{GN}$-secure,
because performing $h$ as first action in a run makes it impossible for $L$ to observe the view $0 \ell 1$.

In \emph{deterministic} systems, for the two-domain policy $H \not \pol L$, the notions $\mname{NDI}_{pw}$ and $\mname{GN}_{pw}$, 
and Goguen and Meseguer's orginal notion of Noninterference are  known to be equivalent. Thus, both  $\mname{NDI}_{pw}$ and $\mname{GN}_{pw}$
are reasonable candidates for the generalization of Noninterference to nondeterministic systems.

\section{Motivation for Abstraction}
\label{sec:abstraction}

The definitions $\mname{NDI}_{pw}$ and $\mname{GN}_{pw}$  have generalized the corresponding definitions $\mname{NDI}$ and $\mname{GN}$  usually given 
for the two-domain policy $H \not \pol L$ in a \emph{pointwise} fashion, stating in different ways that 
there should not be a flow of information from domain $u$ to domain $v$ when $u \not \pol v$. 
We now present some examples that suggest that these pointwise definitions may be weaker than required in the case of 
policies with more than two domains. 

We first present an example which demonstrates that $\mname{NDI}_{pw}$-security
is flawed with respect to combined behaviour of multiple domains.
(Interestingly, this can already be shown in a deterministic system.) 

\begin{example}
  \vspace{-1cm}
  \label{ex:NDI-fruitfly}
    \begin{figure}[h]
  \centering
  \begin{minipage}{0.4\textwidth}
    \begin{tikzpicture}[->, auto]
      \node at (-1, 0) (invis) {};
      \node at (0, 0) (sI) {$s_I$};
      \node at (2, 0) (s0) {$s_0$};
      \node at (4, 0) (s1) {$s_1$};

      \node at (-1, -0.25) () {\tiny $L$:};
      \node at (-1, -0.5) () {\tiny $H_1$:};
      \node at (-1, -0.75) () {\tiny $H_2$:};
      
      \node at (0, -0.25) () {\tiny 0};
      \node at (0, -0.5) () {\tiny $\bot$};
      \node at (0, -0.75) () {\tiny $\bot$};

      \node at (2, -0.25) () {\tiny 0};
      \node at (2, -0.5) () {\tiny $\bot$};
      \node at (2, -0.75) () {\tiny $\bot$};

      \node at (4, -0.25) () {\tiny 1};
      \node at (4, -0.5) () {\tiny $\bot$};
      \node at (4, -0.75) () {\tiny $\bot$};

      \path (invis) edge node {} (sI);
      \path (sI) edge node {${\tiny h_1,h_2}$} (s0);
      \path (s0) edge node {${\tiny \ell}$} (s1);
      \path (sI) edge [loop above] node {${\tiny \ell}$} (sI);
      \path (s0) edge [loop above] node {${\tiny h_1,h_2}$} (s0);
    \end{tikzpicture}
  \end{minipage}
  \begin{minipage}{0.4\textwidth}
    \centering
    \begin{tikzpicture}[|->, auto]
      \node at (0,0) (L) {$L$};
      \node at (-1.5,-0.6) (H1) {$H_1$};
      \node at (-1.5,0.6) (H2) {$H_2$};
      \path (L) edge node {} (H1);
      \path (L) edge node {} (H2);
    \end{tikzpicture}
  \end{minipage}
  \caption{A system and policy showing a weakness of $\mname{NDI}_{pw}$.}
  \label{fig:NDIn-is-bad}
  \vspace{-0.5cm}
  \end{figure}
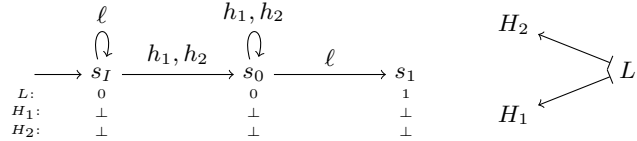
  Consider the system and policy depicted in Figure \ref{fig:NDIn-is-bad}.
  The domain assignment is given by $\dom(l)=L$, $\dom(h_1)=H_1$ and $\dom(h_2)=H_2$.
  We have $H_1 \not\pol L$ and $H_2 \not\pol L$ and show that $H_1 \not\flows_I L$ and $H_2 \not\flows_I L$ hold.
  Let $\alpha = {h_1}^a$ for $a \geq 0$ and $\beta$ be an $L$ view, then $\beta$ must have the form $0 (\ell 0)^b (\ell 1)^c$, where $b,c \geq 0$.
  Consider the run $r = s_I (\ell s_I)^b h_2 s_0 (h_1 s_0)^a (\ell s_1)^c$, which satisfies $\view_L(r) = 0 (\ell 0)^b (\ell 1)^c = \beta$ and $\act_{H_1}(r) = {h_1}^a$,
  and thus $\alpha$ and $\beta$ are $L$ compatible.
  Due to symmetry, we also get $H_2 \not\flows_I L$ with the same argument.
  The system therefore is $\mname{NDI}_{pw}$-secure for the policy.
  However, if $L$ observes the view $0 \ell 1$ then $H_1$ or $H_2$ must have performed $h_1$ or $h_2$, respectively.
  \qed
\end{example}

In the example, domain $L$ cannot know which of $H_1$ or $H_2$ was active upon observing the view $0 \ell 1$,
but $L$ can tell that at least one of them was active nonetheless.
It can be argued that this is a flow of information that is not permitted by the depicted policy.
The example would turn formally insecure if we changed the policy to $H \not\pol L$ and set $\dom(h_1) = \dom(h_2) = H$.
The problem arises as soon as more than one domain must be noninterfering with $L$.

One way to address this weakness of $\mname{NDI}_{pw}$ is to revise the definition so that it deals 
with what a domain can learn about the actions of a set of domains collectively, rather than about these domains individually. 
We may extend the relation $\flows_I$ to sets of domains as follows:
for $X \subseteq D$, $X \neq \emptyset$ and $u \in D$, write $X \flows_I u$ if there are $\alpha \in {A_X}^*$ and $\beta \in \Views_u(\struct{M})$ such that no $r \in \Runs(\struct{M})$ satisfies both $\act_X(r)=\alpha$ and $\view_u(r)=\beta$.
Applying this with the set $X={}^{\not\pol}u$ consisting of all domains that may not interfere with  domain $u$, we obtain the following setwise 
version of Nondeducibility on Inputs:  

\begin{definition}
\label{def:NDI}
A system is \emph{$\mname{NDI}_{sw}$-secure} for $\pol$ if for all $u \in D$, we have that ${}^{\not\pol}u \not\flows_I u$.
\end{definition}

This gives a notion that is intermediate between the pointwise versions of Generalized Noninterference and Nondeducibility on Inputs: 

\begin{proposition}
  \label{thm:GN-in-NDI}
  \label{thm:NDI-strictly-contains-NDIN}
  $\mname{GN}_{pw}$ is strictly contained in $\mname{NDI}_{sw}$,
  and $\mname{NDI}_{sw}$ is strictly contained in $\mname{NDI}_{pw}$.
  A system is $\mname{NDI}_{sw}$-secure for $H \not\pol L$ if and only if it is $\mname{NDI}_{pw}$-secure for $H \not\pol L$.
\end{proposition}

We remark that there is not a need to give a similar setwise definition of Generalized Noninterference, because
the definition of $\mname{GN}_{pw}$ already allows the set of actions in a run to be
modified,  without change to the view of $u$, by arbitrary insertions and 
deletions of actions with domains $v$ in ${}^{\not\pol}u$, through a sequence of 
applications of $GN^+(v,u)$ and $GN^-(v,u)$. 

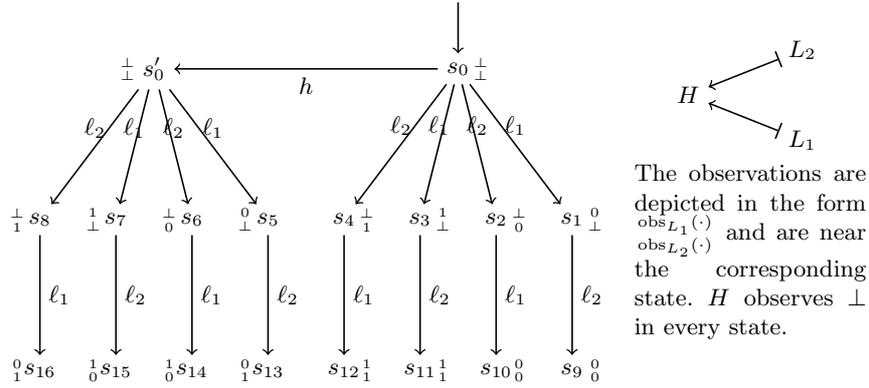
\begin{figure}[t]
  \vspace{-0.5cm}
  \begin{minipage}{0.5\textwidth}
    \centering
    \begin{tikzpicture}[->, semithick, auto, rotate=270]
      \node at (-1,0) (invis) {};
      \node at (0,0) (s0) {$s_0$};
      \path (invis) edge node {} (s0);
      \node at (0, 0.3) () {\tiny ${\bot \atop \bot}$};

      
      \node at (2,+1.5) (s1) {$s_1$};
      \node at (2,+0.5) (s2) {$s_2$};
      \node at (2,-0.5) (s3) {$s_3$};
      \node at (2,-1.5) (s4) {$s_4$};
      \path (s0)
      edge node[above] {$\ell_1$} (s1)
      edge node[above] {$\ell_2$} (s2)
      edge node[above] {$\ell_1$} (s3)
      edge node[above] {$\ell_2$} (s4);
      \node at (2, +1.5+0.3) () {\tiny ${0 \atop \bot}$};
      \node at (2, +0.5+0.3) () {\tiny ${\bot \atop 0}$};
      \node at (2, -0.5+0.3) () {\tiny ${1 \atop \bot}$};
      \node at (2, -1.5+0.3) () {\tiny ${\bot \atop 1}$};

      \node at (4,+1.5) (s9) {$s_9$};
      \node at (4,+0.5) (s10) {$s_{10}$};
      \node at (4,-0.5) (s11) {$s_{11}$};
      \node at (4,-1.5) (s12) {$s_{12}$};
      \path (s1) edge node {$\ell_2$} (s9);
      \path (s2) edge node {$\ell_1$} (s10);
      \path (s3) edge node {$\ell_2$} (s11);
      \path (s4) edge node {$\ell_1$} (s12);
      \node at (4, +1.5+0.3) () {\tiny ${0 \atop 0}$};
      \node at (4, +0.5+0.3) () {\tiny ${0 \atop 0}$};
      \node at (4, -0.5+0.3) () {\tiny ${1 \atop 1}$};
      \node at (4, -1.5+0.3) () {\tiny ${1 \atop 1}$};


      \node at(0,-4) (s0_) {$s_0'$};
      \path (s0) edge node {$h$} (s0_);
      \node at (0, -4-0.35) () {\tiny ${\bot \atop \bot}$};

      \node at (2,-2.5) (s5) {$s_5$};
      \node at (2,-3.5) (s6) {$s_6$};
      \node at (2,-4.5) (s7) {$s_7$};
      \node at (2,-5.5) (s8) {$s_8$};
      \path (s0_)
      edge node[above] {$\ell_1$} (s5)
      edge node[above] {$\ell_2$} (s6)
      edge node[above] {$\ell_1$} (s7)
      edge node[above] {$\ell_2$} (s8);
      \node at (2, -2.5-0.3) () {\tiny ${0 \atop \bot}$};
      \node at (2, -3.5-0.3) () {\tiny ${\bot \atop 0}$};
      \node at (2, -4.5-0.3) () {\tiny ${1 \atop \bot}$};
      \node at (2, -5.5-0.3) () {\tiny ${\bot \atop 1}$};

      \node at (4,-2.5) (s13) {$s_{13}$};
      \node at (4,-3.5) (s14) {$s_{14}$};
      \node at (4,-4.5) (s15) {$s_{15}$};
      \node at (4,-5.5) (s16) {$s_{16}$};
      \path (s5) edge node {$\ell_2$} (s13);
      \path (s6) edge node {$\ell_1$} (s14);
      \path (s7) edge node {$\ell_2$} (s15);
      \path (s8) edge node {$\ell_1$} (s16);
      \node at (4, -2.5-0.3) () {\tiny ${0 \atop 1}$};
      \node at (4, -3.5-0.3) () {\tiny ${1 \atop 0}$};
      \node at (4, -4.5-0.3) () {\tiny ${1 \atop 0}$};
      \node at (4, -5.5-0.3) () {\tiny ${0 \atop 1}$};

    \end{tikzpicture}
  \end{minipage}
  \hspace{2cm}
  \begin{minipage}{0.25\textwidth}
    \centering
    \begin{tikzpicture}[|->, semithick, auto]
      \node at (0,0) (H) {$H$};
      \node at (+1.5,-0.6) (L1) {$L_1$};
      \node at (+1.5,+0.6) (L2) {$L_2$};
      \path (L1) edge node {} (H);
      \path (L2) edge node {} (H);
    \end{tikzpicture}

    \begin{tabular}{p{\textwidth}}
    The observations are depicted in the form ${\obs_{L_1}(\cdot) \atop \obs_{L_2}(\cdot)}$ and are near the corresponding state.
    $H$ observes $\bot$ in every state.
    \end{tabular}
  \end{minipage}
  \caption{System and policy illustrating a collusion attack.}
  \label{fig:GN-fruitfly}
    \vspace{-0.5cm}
\end{figure}

Despite $\mname{NDI}_{sw}$ and $\mname{GN}_{pw}$ being suitable for the multi-domain case and
the latter notion being quite strict, one can argue that neither of them can handle 
collusion,
where multiple domains join forces in order to attack the system as a team.
The system depicted in Figure \ref{fig:GN-fruitfly}, 
a variant of Example 3 and Figure 4 from \cite{RonZhangEngelhardt2012},
 can be shown to 
satisfy $\mname{GN}_{pw}$-security, hence is secure in the strongest sense introduced so far. 
However, if $L_1$ and $L_2$ collude, they can infer from the parity of their observations
that $H$ performed $h$ at the beginning of the run.
This motivates the introduction of 
stronger
\emph{coalition-aware} notions of security.


\section{Reduction-based Notions of Noninterference for Multi-domain Policies}
\label{sec:reduction}

The examples of the previous section indicate that in nondeterministic settings, it is necessary
to deal with groups of agents both on the side of the attackers and the side of the domains being 
attacked. Policy cuts provide types of groupings and enable a reduction to a basic 
notion of security for two-domain policies. The question that then remains is what types
of cut we should use, and which basic notion of security. 
In this section, we define three types of cut and the resulting notions of security when 
$\mname{GN}$ and $\mname{NDI}$ are taken to be the basic notion of security. 

Let $D$ be a set of domains.
For $u \in D$ we define the following two special cuts $\absUp{u}{}$ and $\absDown{u}{}$.
\begin{center}
  \begin{tabular}{ccc}
    $\absUp{u}{} := (u^{\pol}, ~D \setminus u^{\pol})$
    &
    { and }
    &
    $\absDown{u}{} := (D \setminus {}^{\pol} u, ~ {}^{\pol} u)$
  \end{tabular}
\end{center}
The term $\absUp{u}{}$ stands for the cut that forms a High-up coalition starting at domain $u$,
while $\absDown{u}{}$ stands for the cut that forms a Low-down coalition with respect to $u$.
Figure \ref{fig:cut-abstractions} depicts an example of each on the same policy.

\begin{figure}[H]
  \vspace{-0.5cm}
  \begin{minipage}{0.5\textwidth}
    \centering
    \begin{tikzpicture}[|->, semithick, auto, scale=1.25]

      \node at (0,0) (A) {A};
      \node at (1,0) (B) {B};
      \node at (0.5,1) (C) {C};
      \node at (1,2) (D) {D};
      \node at (1.5,1) (E) {E};

      \draw [dotted] plot [smooth cycle] coordinates {
        ($(C)+(-0.3,-0.2)$)
        ($(C)+(+0.2,-0.2)$)
        ($(D)+(+0.3,+0.2)$)
        ($(D)+(-0.2,+0.2)$)
      };
      \draw[dotted] plot [smooth cycle] coordinates {
        ($(A) + (-0.2, -0.2)$)
        ($(A) + (-0.2, +0.2)$)
        ($(B) + (-0.125, +0.3)$)
        ($(E) + (-0.2, +0.2)$)
        ($(E) + (+0.3, +0.2)$)
        ($(B) + (+0.2, -0.2)$)
      };
      \node at (0.25, 1.5) (H) {\tiny ${\cal H}$};
      \node at (1.7, 0.25) (L) {\tiny ${\cal L}$};
   
      \path (A) edge node {} (C)
      (B) edge node {} (C)
      (C) edge node {} (D);

    \end{tikzpicture}
  \end{minipage}
  \begin{minipage}{0.5\textwidth}
    \centering
    \begin{tikzpicture}[|->, semithick, auto, scale=1.25]
      \node at (0,0) (A) {A};
      \node at (1,0) (B) {B};
      \node at (0.5,1) (C) {C};
      \node at (1,2) (D) {D};
      \node at (1.5,1) (E) {E};

      \draw [dotted] plot [smooth cycle] coordinates {
        ($(A)+(-0.2,-0.2)$)
        ($(B)+(+0.2,-0.2)$)
        ($(C)+(0,0.35)$)
      };
      \draw [dotted] plot [smooth cycle] coordinates {
        ($(D)+(-0.3,+0.2)$)
        ($(D)+(+0.2,+0.2)$)
        ($(E)+(+0.3,-0.2)$)
        ($(E)+(-0.2,-0.2)$)
      };
      \node at (1.6, 1.75) (H) {\tiny ${\cal H}$};
      \node at (-0.25, 0.5) (L) {\tiny ${\cal L}$};
      
      \path (A) edge node {} (C)
      (B) edge node {} (C)
      (C) edge node {} (D);
    \end{tikzpicture}
      
  \end{minipage}
  \caption{Cuts $\absUp{C}{}$ and $\absDown{C}{}$ visualized.}
  \label{fig:cut-abstractions}
  \vspace{-0.5cm}
\end{figure}

Abstractions of type $\absUp{\cdot}{}$ are suggested by Ryan
(as discussed in the introduction), 
while the type $\absDown{\cdot}{}$ is what we referred to as its dual.
As already noted in the introduction, there are additional `cut' abstractions that are 
neither High-up nor Low-down. 
In a systematic way, we can now obtain new notions of security 
based on 
cuts as follows.

\begin{definition}
  \label{def:abstraction-definitions}
  Let $\struct{M}$ be a system with domain set $D$ and $\pol$ be a policy over $D$.
  For $X \in \set{\mname{GN}, \mname{NDI}}$, we say $\struct{M}$ is
  \begin{itemize}
  \item 
  Cut $X$-secure ($\mname{C{-}X}$-secure) for $\pol$,
    if $\struct{M}^\absn$ is $X$-secure for $\pol^\absn$ for all  cuts  $\absn$ of $D$,
    
  \item High-up $X$-secure ($\mname{H{-}X}$-secure) for $\pol$,
    if $\struct{M}^{\absUp{u}{H}}$ is $X$-secure for $\pol^{\absUp{u}{H}}$ for all $u \in D$,
    
  \item Low-down $X$-secure ($\mname{L{-}X}$-secure) for $\pol$,
    if $\struct{M}^{\absDown{u}{L}}$ is $X$-secure for $\pol^{\absDown{u}{L}}$ for all $u \in D$.
  \end{itemize}
\end{definition}

There is a straightforward relationship between these notions of $\mname{GN}$ and their $\mname{NDI}$-counterparts.
\begin{proposition}
  \label{thm:straightforward-relationships}
  For all $X \in \set{C, H, L}$:
  the notion $\mname{X{-}GN}$ is strictly contained in $\mname{X{-}NDI}$.
\end{proposition}
This follows directly from Definition \ref{def:abstraction-definitions},
the fact that $\mname{GN}$ implies $\mname{NDI}$ due to Proposition \ref{thm:GN-in-NDI}, and
that the system depicted in Figure \ref{fig:NDI-weaker-than-GN} provides separation for each case.
Also, one would expect that reasonable extensions of $\mname{GN}$ and $\mname{NDI}$ agree if applied to $H \not\pol L$,
and this is exactly what we find, since we can identify singleton coalitions with their only member.

\section{Main Result}
\label{sec:main}
We now state the main result of the paper. 
We have a set of definitions of security that address the need to consider groupings
of attackers and defenders in multi-domain policies, based on two basic notions of 
security $\mname{NDI}$ and $\mname{GN}$ for the two-domain case.  
We are now interested in understanding the relationships between these definitions. 
Additionally, we are interested in understanding which definitions satisfy the desirable 
property of monotonicity. 

\begin{theorem}
  \label{thm:main}
  The notions of 
  $\mname{GN}_{pw}$, 
  $\mname{L{-}GN}$,
  $\mname{H{-}GN}$,
  $\mname{C{-}GN}$,
  $\mname{NDI}_{pw}$,
  $\mname{NDI}_{sw}$,
  $\mname{L{-}NDI}$,
  $\mname{H{-}NDI}$ and
  $\mname{C{-}NDI}$-security
  are ordered by implication as depicted in Figure \ref{fig:main-summary}.
  The containment relations are strict;
  arrows due to reflexivity or transitivity are omitted.
  The name of a notion is underlined if and only if it is monotonic.
\end{theorem}

In particular, we find 
for the GN-variants that 
Ryan's proposal to use reductions based on High-up coalitions 
is complete, in the sense that it yields the same notion of security as a quantification over
all cuts. This notion is moreover adequate in the sense of being monotonic. Somewhat surprisingly, 
the dual notion based on Low-down coalitions is strictly weaker, 
and also fails to be monotonic.  

The situation is different for the basic notion of $\mname{NDI}$.
In this case, we see that Ryan's proposal is not complete with respect to quantification over all cuts. 
Indeed, the resulting notion $\mname{H{-}NDI}$ does not even imply the more adequate setwise version of 
of $\mname{NDI}$, although it does imply the pointwise version. 
The Low-down version of NDI does imply the setwise version, and is independent of $\mname{H{-}NDI}$. 
However, neither $\mname{H{-}NDI}$ nor $\mname{L{-}NDI}$ is monotonic. 
This leaves the (monotonic) cut based variant as the most satisfactory notion in this case.

\begin{figure}[t]
  \vspace{-0.5cm}
  \centering
  \begin{minipage}[t]{\textwidth}
    \centering
    \begin{tikzpicture}[->, auto, node distance=1cm]
    \node at (0,0) (hgn) {$\underline{\mname{H{-}GN}}$};
    \node at (2,0) (cgn) {$\underline{\mname{C{-}GN}}$};
    \node at (1,1) (lgn) {$\mname{L{-}GN}$};
    \node at (1,2) (gn) {$\underline{\mname{GN}_{pw}}$};

    \node[right of=hgn] (eq1) {=};
    
    \path (cgn) edge node {} (lgn);
    \path (hgn) edge node {} (lgn);

    \path (lgn) edge node {} (gn);
    

    \node at (6+0,0) (cndi) {$\underline{\mname{C{-}NDI}}$};
    \node at (6+-1,1) (lndi) {$\mname{L{-}NDI}$};
    \node at (6+-1,2) (ndi) {$\underline{\mname{NDI}_{sw}}$};
    \node at (6+0,3) (ndii) {$\underline{\mname{NDI}_{pw}}$};
    \node at (6+1,1.5) (hndi) {$\mname{H{-}NDI}$};

    \path (cndi) edge node {} (lndi)
    (lndi) edge node {} (ndi)
    (ndi) edge node {} (ndii)
    (cndi) edge node {} (hndi)
    (hndi) edge node {} (ndii)

    (gn) edge node {} (ndi)
    (lgn) edge node {} (lndi)
    (cgn) edge node {} (cndi)
    ;

  \end{tikzpicture}
  \end{minipage}
  \vspace{-0.5cm}
  \caption{Implications between our notions of security.}
  \label{fig:main-summary}
\end{figure}
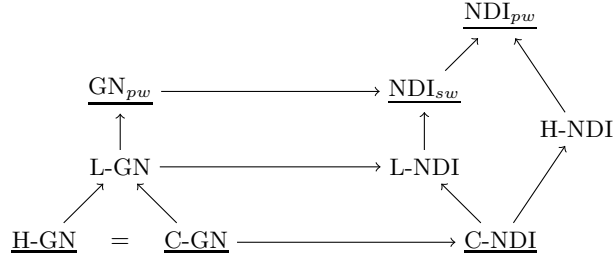

\iftoggle{long}{
  How one arrives at these results is explained in the next section, which gives proofs for all results in this paper.
}{
}

\section{Technical Details}
\label{sec:proofs}

\iftoggle{long}{
\subsection{Relationship between $\mname{GN}_{pw}$ and $\mname{NDI}_{sw}$}

\begin{proof}[of Proposition \ref{thm:GN-in-NDI}]
  \emph{$\mname{GN}_{pw}$ is strictly contained in $\mname{NDI}_{sw}$:}
  For containment, assume a system to be $\mname{GN}_{pw}$-secure for a policy $\pol$,
  let $u$ be a domain and $X$ be a set of domains in $\struct{M}$, where $x \not\pol u$ for all $x \in X$.
  Then the conditions $\mname{GN}^+(x,u)$ and $\mname{GN}^-(x,u)$ guarantee that at any position in the run, any action from $A_X$
  can be inserted or removed `without changing the $u$ view of the run'. More precisely,
  one always finds a run with the same $u$ view and the desired action inserted or removed at any position.
  Therefore, all $u$ views are compatible with all sequences from ${A_X}^*$ and the system is $\mname{NDI}_{pw}$-secure for $\pol$.
  Separation is due to the system depicted in Figure \ref{fig:NDI-weaker-than-GN}
  and the fact that $\mname{NDI}_{\mathit{sw}}$ and $\mname{NDI}$ are equivalent on $H \not\pol L$.

  \emph{$\mname{NDI}_{sw}$ is strictly contained in $\mname{NDI}_{pw}$:}
  Containment is clear, since $\mname{NDI}_{pw}$-security is $\mname{NDI}$-security restricted to the case of singletons
  and thus follows directly from $\mname{NDI}$.

  The system in Example \ref{ex:NDI-fruitfly} separates $\mname{NDI}_{sw}$ and $\mname{NDI}_{pw}$.
  For $\mname{NDI}_{pw}$, it suffices to test if $H_1 \not\flows_I L$ because of symmetry.
  We have $H_1 \not\flows_I L$ because by visiting $s_0$ an appropriate number of times we can add any number of $h_1$ actions to a run without changing its $L$ view.
  As already seen, this system is not $\mname{NDI}$-secure;
  if $L$ observes the view $0 \ell 1$ the action sequence $\varepsilon \in \set{h_1,h_2}^*$ wasn't performed by $\set{H_1,H_2}$.
  \qed
\end{proof}



  


\subsection{Relationships between Cut-based Notions of $\mname{GN}$}

}{ 
  
  This section provides an overview of how to show some relationships claimed by
  Theorem \ref{thm:main}.
  First, the concept of a \emph{vulnerability} of a system is introduced for $\mname{GN}_{pw}$ and $\mname{NDI}_{sw}$.
  A vulnerability of a system is a witness of a security violation for the respective notion of security.
  The proofs to compare the different variants of $\mname{GN}$ and $\mname{NDI}$ are done by contraposition
  and show how, for given cuts $\absn_0$ and $\absn_1$, a vulnerability of $\struct{M}^{\absn_0}$ can be translated into a vulnerability of $\struct{M}^{\absn_1}$.
  This shows that security for all cuts of type $\absn_1$ implies security for all cuts of type $\absn_0$.
  We use technical lemmas, one for $\mname{GN}_{pw}$ and one for $\mname{NDI}_{sw}$, that give sufficient conditions in order to
  facilitate this translation.
  Separation of two notions is done by giving a concrete example that exhibits the one, but not the other,
  security property.

  Due to space constraints, we only present proofs for some relationships between $\mname{GN}$ variants.
  Proofs for all results in this section can be found in \cite{longVersion}.

  \paragraph{$\mname{GN}$ vulnerabilities.}
}
A \iftoggle{long}{\emph{$\mname{GN}$ vulnerability}}{\emph{vulnerability}}  
of a system $\struct{M}$ is a tuple $(u,\alpha_0,a,\alpha_1,\beta,\pol)$, where
$\pol$ is a policy over the domain set of $\struct{M}$,
$u$ is a domain in $\struct{M}$,
$\alpha_0,\alpha_1 \in A^*$,
$a \in A$ with $\dom(a) \not\pol u$ and
$\beta \in \Views_u(\struct{M})$ such that there is a run $r$ that satisfies
$\view_u(r) = \beta$ and at least one of
\begin{itemize}
\item $\act(r) = \alpha_0 \alpha_1$ and no run $r'$ with $\act(r') = \alpha_0 a \alpha_1$ satisfies $\view_u(r') = \beta$,
\item $\act(r) = \alpha_0 a \alpha_1$ and no run $r'$ with $\act(r') = \alpha_0 \alpha_1$ satisfies $\view_u(r') = \beta$.
\end{itemize}

\iftoggle{long}{If the context is clear, we only say \emph{vulnerability}.}{}  
We evidently have a vulnerability of a system if and only if $\mname{GN}^+(\dom(a),u)$ or $\mname{GN}^-(\dom(a),u)$ does
not hold.
Without loss of generality we always assume a violation of $\mname{GN}^+(\dom(a),u)$ if there is a vulnerability,
since the case of a $\mname{GN}^-(\dom(a),u)$ violation is 
similar. 

\ifboolexpr{not togl {long}} { 
  \subsection{Translating $\mname{GN}$ Vulnerabilities}
}

\iftoggle{long}{
The proofs to compare the different cut-based variants of $\mname{GN}$ are done by contraposition
and show how, for given cuts $\absn_0$ and $\absn_1$, a vulnerability of $\struct{M}^{\absn_1}$ can be
translated into a vulnerability of $\struct{M}^{\absn_0}$.}{} 
The next definition formalizes the idea that 
the view of 
an attacking coalition,
e.g. the Low domain of a cut,
has at least as much information as the view of a sub-coalition. 
We will need this to argue that if a coalition possesses enough information to successfully launch an
attack on a system (i.e. it can violate $\mname{GN}^+$)
then, a fortiori, a bigger coalition possesses enough information for an attack.

\begin{definition}
  \label{def:projection}
  Let
  $\struct{M}$ be a system with action set $A$ and observation set $O$,
  let ${\cal D}_0$ and ${\cal D}_1$ be abstractions of its domain set, and
  $F \in {\cal D}_0$,
  $G \in {\cal D}_1$ such that $F \subseteq G$.
  Then the operator
  \[
  \proj{F}{G} \colon A \cup O^G \cup \Views_G(\struct{M}^{{\cal D}_1}) \to \Views_F(\struct{M}^{{\cal D}_0})
  \]
  is defined as follows:
    
  \begin{itemize}
  \item if $a \in A$ then $\proj{F}{G}(a) = a|_F$,
    where $a$ is considered to be a sequence of length one.
    The result is its subsequence of actions that $F$ can perform, i.e. it is either $a$ or $\varepsilon$,
    
  \item if $o \in O^G$ then $\proj{F}{G}(o) = o|_F$,
    that is observations made by $G$ are restricted such that the result is the observation made by $F$,

  \item if $\alpha$ is a $G$ view and $\beta \in O^G \cup A_G \cdot O^G$ such that $\alpha\beta$ is a $G$ view, then
    \[
    \proj{F}{G}(\alpha\beta) = \proj{F}{G}(\alpha) \aconc \proj{F}{G}(\beta).
    \]

  \end{itemize}
  For all other cases let the result be undefined.
  The symbol $\proj{F}{G}(\cdot)$ is chosen to support the intuition that $G$ views are `projected down' to $F$ views.
\end{definition}

That the previous definition is reasonable is established by a correctness lemma which makes the restriction aspect of the operator clear.

\begin{lemma}
  \label{thm:projection-lemma}
  Let $\struct{M}$ be a system,
  let ${\cal D}_0$ and ${\cal D}_1$ be abstractions of its domain set,
  and $F \in {\cal D}_0$,
  $G \in {\cal D}_1$ such that $F \subseteq G$.
  Then for all $r \in \Runs(\struct{M})$ we have
  $\proj{F}{G}(\view_G(r)) = \view_F(r)$.
\end{lemma}
\iftoggle{long}{
\begin{proof}
  First, we confine ourselves to case of observations.
  Let $o \in O^G$ such that $\obs^{\absn_1}_G(s)=o$ for some state $s$,
  then $o$ is a function that maps each $u \in G$ to an element in $O$.
  The function $\proj{F}{G}(o) = o|_{F \times O}$ has the domain set $F$,
  and is total since $o$ is total and we have $F \subseteq G$.
  We get that $\proj{F}{G}(o)$ is the observation of domain $F$ made in state $s$ and is thus equal to $\obs_F^{\absn_0}(s)$.
  
  The main result is shown by induction over runs.
  The base case follows from the previous paragraph.
  For the induction step, let $r$ be a run of the form $r' as$, where $r \in \Runs(\struct{M})$, $a$ is an action and $s$ a state of $\struct{M}$.
  We distinguish two cases.
  
  If $\dom^{\absn_1}(a) \neq G$
    then we have $\view_G(r as') = \view_G(r) \aconc \obs_G^{\absn_1}(s')$ and
    $\proj{F}{G}(\view_G(r) \aconc \obs_G^{\absn_1}(s')) = \proj{F}{G}(\view_G(r)) \aconc \proj{F}{G}(\obs_G^{\absn_0}(s'))$,
    and this is equal to
    $\view_F(r) \aconc \obs_F^{\absn_0}(s')$ by induction and the special case of observations above.
    This value equals $\view_F(r as')$ by definition of $\view$.
    Note that we must have $\dom^{\absn_0}(a) \neq F$,
    for if we had $\dom^{\absn_0}(a) = F$ then $\dom(a) \in F$, which implies $\dom(a) \in G$ and yields $\dom^{\absn_1}(a) = G$, contrary to the case assumption.

  If $\dom^{\absn_1}(a) = G$
    then $\view_G(r as') = \view_G(r) \cdot a \cdot \obs_G^{\absn_1}(s')$.
    Applying $\proj{F}{G}$ and induction yields
    $\proj{F}{G}(\view_G(r as')) = \view_F(r) \aconc \proj{F}{G}(a) \aconc \obs_F^{\absn_0}(s')$,
    which is equal to $\view_F(r as')$ in both cases $\dom^{\absn_0}(a) = F$ and $\dom^{\absn_0}(a) \neq F$.
  \qed
\end{proof}
}{} 

Conditions under which the translation of a vulnerability is possible are established by the following result:
the attacking coalition may not shrink
and the translation must respect the status of being the attacker's victim.

\begin{lemma}
  \label{thm:GN-vuln-lemma}
  Let $\struct{M}$ be a system and $\pol$ be a policy, 
  and $\absn_0 = ({\cal H}_0,{\cal L}_0)$ be a cut of the domain set of $\struct{M}$
  with respect to $\pol$. 
  Let $(F, \alpha_0, a, \alpha_1, \beta, \pol^{\absn_0})$ be a vulnerability of $\struct{M}^{\absn_0}$.
  Let $\absn_1 = ({\cal H}_1, {\cal L}_1)$ be a cut such that there is $G \in \set{{\cal H}_1,{\cal L}_1}$ with
  $\dom^{\absn_1}(a) \not\pol^{\absn_1} G$ and $F \subseteq G$.
  Then there is $\beta' \in \Views_G(\struct{M}^{\absn_1})$ such that a vulnerability of $\struct{M}^{\absn_1}$
  is given by $(G, \alpha_0, a, \alpha_1, \beta', \pol^{\absn_1})$.
\end{lemma}
\iftoggle{long}{
\begin{proof}
Since $(F, \alpha_0, a, \alpha_1, \beta, \pol^{\absn_0})$ is a vulnerability of $\struct{M}^{\absn_0}$,
there is a run $r$ on $\alpha_0 \alpha_1$ and an $F$ view of $\beta$ such that no run on $\alpha_0 a \alpha_1$ attains the $F$ view $\beta$.
Due to the prerequisites it suffices to show that there is a $G$ view $\beta'$ attained by some run on $\alpha_0 \alpha_1$
such that no run on $\alpha_0 a \alpha_1$ can attain a $G$ view of $\beta'$,
because then $(G, \alpha_0, a, \alpha_1, \beta', \pol^{\absn_1})$ is a vulnerability
of $\struct{M}^{\absn_1}$,
due to a violation of $\mname{GN}^+(\dom^{\absn_1}(a), G)$,
and we are finished.
By vulnerability, there is a run on $\alpha_0 \alpha_1$ with $F$ view of $\beta$;
let $\beta'$ be $G$ view of that run.
If there were a run $r$ on $\alpha_0 a \alpha_1$ with $\view_{G}(r) = \beta'$ then
this run would satisfy $\view_F(r) = \proj{F}{G}(\view_{G}(r)) = \proj{F}{G}(\beta') = \beta$ by the same lemma,
which contradicts the existence of the vulnerability of $\struct{M}^{\absn_0}$.
Therefore, no such run can exist and we have found a vulnerability of $\struct{M}^{\absn_1}$ as claimed.
\qed
\end{proof}
}{} 

\iftoggle{long}{
Some relationships between cut-based variants of $\mname{GN}$ are trivial and can be seen directly.

\begin{proposition}
  \label{thm:GN-trivial-implications}
  $\mname{C{-}GN}$ implies both $\mname{H{-}GN}$ and $\mname{L{-}GN}$.
\end{proposition}

}{} 

That High-up $\mname{GN}$ implies Cut $\mname{GN}$, and therefore 
these 
notions are equivalent,
might not be apparent, but can be explained by the fact that the Low component $\cal L$ of a cut is
the intersection of the Low components ${\cal L}_0, ..., {\cal L}_{n-1}$ of $n$ High-up cuts
and thus can obtain no more information about High behaviour than each ${\cal L}_i$ individually.
If we can prevent each ${\cal L}_i$ from obtaining any information about how High actions
are interleaved into runs then the same must apply to $\cal L$ as well.

\begin{theorem}
  \label{thm:CGN=HGN}
  The notions $\mname{C{-}GN}$ and $\mname{H{-}GN}$ are equivalent.
\end{theorem}
\begin{proof}
  \iftoggle{long}{
    Because of Proposition \ref{thm:GN-trivial-implications}  it suffices to show that $\mname{H{-}GN}$ implies $\mname{C{-}GN}$.
  }{ 
    Because $\mname{C{-}GN}$ implies $\mname{H{-}GN}$   it suffices to show that $\mname{H{-}GN}$ implies $\mname{C{-}GN}$.
  }
  The proof is done by contraposition and translates a vulnerability with respect to an arbitrary cut
  into a vulnerability with respect to a $\absUp{\cdot}{}$-style cut.

  Let $\struct{M}$ be a system with domain set $D$, $\pol$ a policy over $D$ and $\absn_0$ a cut of $D$.
  Furthermore, let $(F, \alpha_0, a, \alpha_1, \beta, \pol^{\absn_0})$ be a $\mname{GN}$ vulnerability of $\struct{M}^{\absn_0}$.
    Set $\absn_1 := \absUp{\dom(a)}{}$, ${\cal H} := \dom(a)^{\pol}$ and ${\cal L} := D \setminus \dom(a)^{\pol}$.
    Then we have $\absn_1 = ({\cal H},{\cal L})$.
    We show that the prerequisites for Lemma \ref{thm:GN-vuln-lemma} are satisfied, which gives us a vulnerability of $\struct{M}^{\absn_1}$.
    
    First, we demonstrate that $\dom^{\absn_1}(a){=}{\cal H} \not\pol^{\absn_1} {\cal L}$.
    Let $u \in {\cal H}$ and $v \in {\cal L}$, we must show that $u \not\pol v$.
    Assume $u \pol v$, then by choice of $\absn_1$ we have $\dom(a) \pol u$,
    which implies $\dom(a) \pol u \pol v$ and $\dom(a) \pol v$ by transitivity.
    Therefore $v \in {\cal H}$, which contradicts $v \in {\cal L}$, and hence we have $u \not\pol v$.
    It remains to prove that $F \subseteq {\cal L}$.
    Let $u \in F$,
    then due to vulnerability we have $\dom^{\absn_0}(a) \not\pol^{\absn_0} F$,
    i.e. $\dom(a) \not\pol u$.
    By choice of $\absn_1$ we get $u \not\in {\cal H}$, which is equivalent to $u \in {\cal L}$.
    Now application of Lemma \ref{thm:GN-vuln-lemma} yields a vulnerability of $\struct{M}^{\absn_1}$.
  \qed
\end{proof}



\iftoggle{long}{}{ 
\subsection{Separation of the $\mname{GN}$ Variants}
}

The result obtained by Theorem \ref{thm:CGN=HGN} shows completeness of Ryan's technique for $\mname{GN}$.
From this follows that the High-up variant of $\mname{GN}$ implies the Low-down variant.
There is also an example that demonstrates that 
these
notions are distinct, and thus the High-up variant is stricter.

\begin{theorem}
  \label{thm:downcuts-dont-cut-it}
  $\mname{H{-}GN}$ is strictly contained in $\mname{L{-}GN}$.
\end{theorem}
\begin{proof}
  \iftoggle{long}{
    Containment follows from the facts that $\mname{H{-}GN} = \mname{C{-}GN}$ by Theorem \ref{thm:CGN=HGN} and the trivial implications from Proposition \ref{thm:GN-trivial-implications}.
  }{ 
    Containment follows from the facts that $\mname{H{-}GN}=\mname{C{-}GN}$ and that $\mname{C{-}GN}$ implies $\mname{L{-}GN}$.
  }
  For separation, we recall 
  Figure \ref{fig:GN-fruitfly}, 
  and modify it slightly to suit our needs.
This system can be verified to be $\mname{GN}$-secure for the separation policy (i.e., the identity relation) on $\set{H,L_1,L_2}$;
add the edges $(L_1, H)$ and $(L_2, H)$ to it and call it $\pol$.
\iftoggle{long}{
  We anticipate the result that $\mname{GN}_{pw}$ is monotonic (see Proposition \ref{thm:GN-monotonic}),
  and get that the system is $\mname{GN}_{pw}$-secure for $\pol$.
}{ 
  We use that $\mname{GN}_{pw}$ is monotonic,
  and get that the system is $\mname{GN}_{pw}$-secure for $\pol$.
}

With respect to $\pol$, the domain set has two Low-down cuts,
which are $\absDown{L_1}{}$ and $\absDown{L_2}{}$. 
The systems $\struct{M}^{\absDown{L_1}{}}$ and $\struct{M}^{\absDown{L_2}{}}$ can be shown to be $\mname{GN}$-secure for $\pol^{\absDown{L_1}{}}$ and $\pol^{\absDown{L_2}{}}$,
respectively,
and therefore $\struct{M}$ is $\mname{L{-}GN}$-secure for $\pol$.
However, 
for the High-up cut $\absUp{H}{}$,
one can see that $\struct{M}^{\absUp{H}{}}$ fails to be $\mname{GN}$-secure for $\pol^{\absUp{H}{}}$.
Consider the run $r := s_0 h s_0' \ell_1 s_5 \ell_2 s_{13}$.
We have $\view_L(r) = {\bot \atop \bot} \ell_1 {0 \atop \bot} \ell_2 {0 \atop 1}$, where $L$ observations are written in the form ${ \obs_{L_1}(\cdot) \atop \obs_{L_2}(\cdot) }$.
By the parity of their final observations after performing $r$, domains $L_1$ and $L_2$ together can determine that $H$ performed $h$ at the very beginning of the run.
Thus, $\struct{M}^{\absUp{H}{}}$ doesn't satisfy the property $\mname{GN}^-(\set{H}, \set{L_1,L_2})$ for $\pol^{\absUp{H}{}}$,
which means that $\struct{M}$ is not $\mname{H{-}GN}$-secure for $\pol$.
\qed
\end{proof}

The weakness of Low-down $\mname{GN}$ is that it assumes a somewhat restricted attacker
that never groups domains into Low that may not interfere with each other according to the policy.
(For example, for the policy in Figure~\ref{fig:GN-fruitfly}, the coalition $\{L_1,L_2\}$ is not covered.)
But nevertheless such coalitions are possible, which provides an argument against Low-down $\mname{GN}$
if coalitions are a risk.
\iftoggle{long}{
  In a later subsection about monotonicity, we will show that Low-down $\mname{GN}$ is not monotonic,
  which one can interpret as further evidence that it might seem problematic.
  However, Low-down $\mname{GN}$ doesn't break all our intuitions;
  as one might expect, it turns out to be stricter than 
  $\mname{GN}_{pw}$.
}{
  However, as one would expect, Low-down $\mname{GN}$ turns out to be stricter than 
  $\mname{GN}_{pw}$.
}

\begin{theorem}
  \label{thm:GN-neq-downGN}
  $\mname{L{-}GN}$ is strictly contained in $\mname{GN}_{pw}$.
\end{theorem}
\begin{proof}
  Containment is shown by contraposition.
  Let $\struct{M}$ be a system with domain set $D$ and $\pol$ a policy over $D$.
  Assume that $\struct{M}$ is not $\mname{GN}_{pw}$-secure for $\pol$
  and has a vulnerability $(u,\alpha_0,a,\alpha_1,\beta,\pol)$.

  Set $\absn := \absDown{u}{}$, ${\cal L} := {}^{\pol}u$ and ${\cal H} := D \setminus {}^{\pol}u$.
  We show, using Lemma \ref{thm:projection-lemma},
  that there is $\beta'$ so that $({\cal L}, \alpha_0, a, \alpha_1, \beta', \pol^\absn)$
  is a vulnerability in $\struct{M}^\absn$.
  First, we have $\dom^\absn(a) = {\cal H}$, due to $\dom(a) \not\pol u$,
  which implies $\dom^\absn(a) \not\pol^\absn {\cal L}$.
  Next, we demonstrate existence of a suitable $\beta'$.
  We identify observations made by $v$ with observations made by the singleton coalition $\set{v}$,
  and consider the trivial abstraction of $D$, which is $\pset{\set{w}}{w \in D}$.
  Then we clearly have $\set{v} \subseteq {\cal L}$ and can apply Lemma \ref{thm:projection-lemma}.
  Due to vulnerability, there is a run on $\alpha_0\alpha_1$ which has a $\set{u}$ view of $\beta$ such that
  no run on $\alpha_0 a \alpha_1$ has a $\set{u}$ view of $\beta$.
  Let $\beta'$ be the ${\cal L}$ view of this run.
  If there were a run $r$ on $\alpha_0 a \alpha_1$ with ${\cal L}$ view of $\beta'$, then
  $\view_u(r) = \proj{\set{u}}{{\cal L}}(\view_{\cal L}(r)) = \proj{\set{u}}{{\cal L}}(\beta') = \beta$
  by identification of $u$ and $\set{u}$ and Lemma \ref{thm:projection-lemma},
  contradicting the violation of $\mname{GN}^+(u,v)$ in $\struct{M}$.
  Therefore, no such run can exist and
  $({\cal L}, \alpha_0, a, \alpha_1, \beta',\pol^{\absn})$ is a vulnerability of $\struct{M}^\absn$.
  
  For separation,
  take the example from Theorem \ref{thm:downcuts-dont-cut-it} and add the additional edge $(L_1,L_2)$ to $\pol$.
  \iftoggle{long}{
    The system is still $\mname{GN}_{pw}$-secure for $\pol$ due to Proposition \ref{thm:GN-monotonic},
  }{
    The system is still $\mname{GN}_{pw}$-secure for $\pol$, as $\mname{GN}_{pw}$ is monotonic,
  }
  but since we have $\set{H} \not\pol^{\absDown{L_2}{}} \set{L_1,L_2}$,
  the system $\struct{M}^{\absDown{L_2}{}}$ is not $\mname{GN}$-secure by the argument in the proof of Theorem \ref{thm:downcuts-dont-cut-it}.
  \qed
\end{proof}

\iftoggle{long}{
  This concludes our study of cuts in the context of $\mname{GN}$.
}{}

\iftoggle{long}{

\subsection{Relationships between Cut-based Notions of $\mname{NDI}$}

In this subsection, an \emph{$\mname{NDI}$ vulnerability} of a system $\struct{M}$ is a tuple $(u,\alpha,\beta,\pol)$,
where $\pol$ is a policy over the domain set of $\struct{M}$,
$u$ is a domain in $\struct{M}$,
${\dom(a) \not\pol u}$ for all actions $a$ that occur in $\alpha$,
and there is no run $r$ of $\struct{M}$ with $\act_{{}^{\not\pol}u}(r)=\alpha$ and $\view_u(r)=\beta$.
If the context is clear, we only say \emph{vulnerability}.
Clearly, a system is $\mname{NDI}_{sw}$-secure if and only if it has no vulnerabilities.

We will follow the same strategy as used in the previous subsection,
and first provide a lemma to translate vulnerabilities,
then give proofs for the relationships claimed in Theorem \ref{thm:main}.

In order to translate vulnerabilities from one cut to another,
we again must make sure that the attacking coalition doesn't shrink.
Additionally, since $\mname{NDI}_{sw}$ deals with combined behaviour,
the translation must make sure that some noninterference constraints,
which are pairs $(u,v)$ such that $u \not\pol v$,
are preserved.
\begin{lemma}
  \label{thm:NDI-vuln-lemma}
  Let $\struct{M}$ be a system, $\pol$ a policy over its domain set,
  $\absn_0=({\cal H}_0,{\cal L}_0)$ be a cut of its domain set,
  and $F \in \set{{\cal H}_0,{\cal L}_0}$ such that
  $(F,\alpha,\beta,\pol^{\absn_0})$ is a vulnerability of $M^{\absn_0}$.
  Let $\absn_1 = ({\cal H}_1,{\cal L}_1)$ be a cut such that there is $G \in \set{{\cal H}_1,{\cal L}_1}$ with
  \begin{enumerate}
    \item for all actions $a$ that occur in $\alpha$, we have $\dom^{\absn_1}(a) \not\pol^{\absn_1} G$, and
    \item $F \subseteq G$.
  \end{enumerate}
  Then there is $\beta' \in \Views_{G}(\struct{M}^{\absn_1})$ so that
  $(G, \alpha, \beta', \pol^{\absn_1})$ is a vulnerability of $\struct{M}^{\absn_1}$.
\end{lemma}
\begin{proof}
  Due to the prerequisites, it only remains to show the existence of a suitable $\beta'$.
  Let $\beta'$ be a $G$ view with $\proj{F}{G}(\beta') = \beta$.
  Set ${\cal F} := {}^{\not\pol_0}F$ and   
  ${\cal G} := {}^{\not\pol_1}G$.  
  Prerequisite 1 gives us $\alpha \in {A_{\cal G}}^*$.
  If there were a run $r$ of $\struct{M}^{\absn_1}$ with $\act_{\cal G}(r) = \alpha$ and $\view_G(r) = \beta'$,
  then the same run would satisfy $\act_{\cal F}(r) = \alpha$,
  since by vulnerability $\alpha$ consists of actions by domains in ${\cal F}$ only,
  and because we have $\view_F(r) = \proj{F}{G}(\view_{G}(r)) = \proj{F}{G}(\beta') = \beta$ by Lemma \ref{thm:projection-lemma},
  which contradicts the vulnerability of $\struct{M}^{\absn_1}$.
  Therefore, no such run can exist and $(G, \alpha, \beta', \pol^{\absn_1})$ is a vulnerability in $\struct{M}^{\absn_1}$.
  \qed
\end{proof}

Just as with $\mname{GN}$, some relationships are trivial and can be seen from Definition \ref{def:abstraction-definitions} right away.

\begin{proposition}
  \label{thm:NDI-trivial-implications}
  The notion $\mname{C{-}NDI}$ implies $\mname{H{-}NDI}$ and $\mname{L{-}NDI}$.
\end{proposition}

Contrary to $\mname{GN}$, however, where High-up $\mname{GN}$ is strictly contained in Low-down $\mname{GN}$,
we have instead the somewhat surprising situation that the corresponding variants of $\mname{NDI}$ are incomparable.
The next theorem provides the necessary examples.

\begin{theorem}
  \label{thm:LNDI-incomparable-HNDI}
  The notions $\mname{L{-}NDI}$ and $\mname{H{-}NDI}$ are incomparable with respect to implication.
\end{theorem}
\begin{proof}
  \emph{$\mname{H{-}NDI}$ does not imply $\mname{L{-}NDI}$:}
  Consider the system and policy in Example \ref{ex:NDI-fruitfly}.
  In the proof of Proposition \ref{thm:NDI-strictly-contains-NDIN} it is shown that the system violates $\mname{NDI}_{sw}$-security with respect to the cut $\absDown{L}{}$,
  which groups $H_1$ and $H_2$ together.
  It is therefore not $\mname{L{-}NDI}$-secure.
  However, it is $\mname{H{-}NDI}$-secure for the depicted policy.
  To show this, it is enough to prove $\mname{NDI}$-security of the system with respect to $\absUp{H_1}{}$,
  because the case $\absUp{H_2}{}$ is symmetrical to it.

  Set $\absn := \absUp{H_1}{}$, ${\cal H} := \set{H_1}$ and ${\cal L} := \set{H_2,L} \setminus H$.
  Then $\absn = ({\cal H},{\cal L})$ and ${\cal H} \not\pol^\absn {\cal L}$.
  Let $\alpha \in \set{h_1}^*$ and $\beta$ be an ${\cal L}$ view.
  Then $\alpha$ has the form ${h_1}^k$ for $k \geq 0$ and $\beta$ is an element of the language described by one of the regular expressions
  ${0 \atop \bot} (\ell {0 \atop \bot})^n (h_2 {0 \atop \bot})^m$
  for $n,m \geq 0$, or
  ${0 \atop \bot} (\ell  {0 \atop \bot})^n (h_2 {0 \atop \bot})^m \ell {1 \atop \bot} ((l + h_2) {1 \atop \bot})^k$
  for $n\geq 0$, $m \geq 1$ and $k \geq 0$,
  where ${\cal L}$ observations are noted as ${ \obs_L(\cdot) \atop \obs_{H_2}(\cdot) }$.
  It is clear that there is a run that demonstrates the compatibility of $\alpha$ and $\beta$:
  the state $s_1$ can be visited $k$ times for performing $\alpha$.
  We therefore have ${\cal H} \not\flows_I {\cal L}$ and conclude that the system is $\mname{H{-}NDI}$-secure.

  \emph{$\mname{L{-}NDI}$ does not imply $\mname{H{-}NDI}$:}
  Consider the system in Figure \ref{fig:GN-fruitfly}.
  To prove it $\mname{L{-}NDI}$-secure, it suffices to do so for the cut $\absn := \absDown{L_1}{}$,
  because the case of the only other $\absDown{\cdot}{}$ cut is symmetric to it.
  Views perceived by $\set{L_1}$ (here, we identify $\set{L_1}$ with $L_1$) have the form $0 (\ell_1 0)^n$ or $0 (\ell_1 1)^n$ for $n \geq 0$.
  All these views are $\set{L_1}$ compatible with all $\alpha \in \set{l_2, h}^*$,
  because they can be attained by the system performing $l_1$ only,
  or $\alpha$ can be added to a run by looping at states $s_0$, $s_1$, $s_3$, $s_9$ or $s_{11}$.
  The system is therefore $\mname{L{-}NDI}$-secure.
  
  However, it is not $\mname{H{-}NDI}$-secure;
  take the cut $\absUp{H}{}$,
  set ${\cal H} := \set{H}$ and ${\cal L} := \set{L_1,L_2}$,
  and consider the action sequence $\varepsilon$ performed by domain ${\cal H}$.
  The ${\cal L}$ view $\beta := {\bot \atop \bot} \ell_1 {0 \atop \bot} \ell_2 {0 \atop \bot}$ can only be attained if the first action performed in the system is $h$.
  Therefore $\varepsilon$ and $\beta$ are not compatible and we have ${\cal H} \flows_I {\cal L}$.
  \qed
\end{proof}

These results show that Ryan's technique is not `complete' for Nondeducibility on Inputs, as
High-up $\mname{NDI}$ and Low-down $\mname{NDI}$ are incomparable.
The question if High-up $\mname{NDI}$ is complete can now be answered, because if
High-up $\mname{NDI}$ implied Cut $\mname{NDI}$, then High-up $\mname{NDI}$ would
also imply Low-down $\mname{NDI}$ due to Proposition \ref{thm:NDI-trivial-implications},
which would contradict the result from Theorem \ref{thm:LNDI-incomparable-HNDI}.
With symmetry, the same argument holds for High-up $\mname{NDI}$ and Low-down $\mname{NDI}$ swapped,
and we have the following corollary.

\begin{corollary}
  $\mname{C{-}NDI}$ is strictly contained in both $\mname{H{-}NDI}$ and $\mname{L{-}NDI}$.
\end{corollary}

The previous theorem alone doesn't yield evidence in favor of High-up or Low-down.
As with $\mname{GN}$, one might expect the High-up variant to be more adequate,
but it turns out that this isn't the case.
We can argue against High-up using the system presented in Example \ref{ex:NDI-fruitfly}.
As shown in the proof of Theorem \ref{thm:LNDI-incomparable-HNDI}, it is High-up $\mname{NDI}$- but
not Low-down $\mname{NDI}$-secure due to the cut $\absDown{L}{}$ introducing a vulnerability.
The cut $\absDown{L}{}$ aggregates the domains into $\set{L}$ and ${}^{\not\pol}L = \set{H_1, H_2}$.
But if $L$ can infer from observing a certain view that the domains in ${}^{\not\pol}L$ did not perform some action sequence,
this means that the system is not $\mname{NDI}_{sw}$-secure.
In other words, $\mname{NDI}_{sw}$ and Low-down $\mname{NDI}$ are equivalent notions on the example.

\begin{corollary}
  \label{corollary:HNDI-doesnt-imply-NDI}
  $\mname{H{-}NDI}$ does not imply $\mname{NDI}_{sw}$.
\end{corollary}

The weakness of High-up $\mname{NDI}$ is, similar to Low-down $\mname{GN}$, that it doesn't group
domains into High that are incomparable in the policy, while $\mname{NDI}$ does.
The definition of $\mname{NDI}$ is a natural extension of the pointwise application of two-level Nondeducibility on Inputs,
so there is an argument that High-up is not adequate in the setting of $\mname{NDI}$.
The case against it can be made even stronger by proving that Low-down $\mname{NDI}$ does not have this undesirable property.
which is what the next result accomplishes.

\begin{theorem}
  \label{thm:LNDI-implies-NDI}
  $\mname{L{-}NDI}$ implies $\mname{NDI}_{sw}$.
\end{theorem}
\begin{proof}
  Let $\struct{M}$ be a system with domain set $D$ and $\pol$ be a policy over $D$.
  Assume $\struct{M}$ is not $\mname{NDI}_{sw}$-secure for $\pol$,
  then there are $u \in D$, $\alpha \in {A_{{}^{\not\pol}u}}^*$ and $\beta \in \Views_u(\struct{M})$ so that
  $(u,\alpha,\beta,\pol)$ is a vulnerability of $\struct{M}$.
  Clearly, for the abstraction ${\cal D}$ given by $\pset{ \set{v} }{ v \in D }$
  the system $\struct{M}^{\cal D}$ has the vulnerability $(\set{u}, \alpha, \beta, \pol^{\cal D})$.
  
  Consider the cut $\absn := \absDown{u}{}$,
  we use Lemma \ref{thm:NDI-vuln-lemma} to show that there is a vulnerability of $\struct{M}^{\absn}$.
  Set ${\cal L} := {}^{\pol}u$ and ${\cal H} := D \setminus {}^{\pol}u$.
  We obviously have $\set{u} \subseteq {\cal L}$, so it only remains to show that $\alpha$ consists
  only of actions performed by domains that may not interfere with ${\cal L}$ with respect to $\pol^{\absn}$.
  Let $a$ be an action that occurs in $\alpha$.
  Since $\dom(a) \not\pol u$ by vulnerability in $\struct{M}^{\cal D}$ we get that
  $\dom(a) \not\in {\cal L}$ by choice of $\absn$, which means $\dom(a) \in {\cal H}$,
  and this implies $\dom^{\absn}(a) = {\cal H}$ and therefore $\dom^{\absn}(a) \not\pol^{\absn} {\cal L}$.

  Application of Lemma \ref{thm:NDI-vuln-lemma} now gives us a vulnerability of $\struct{M}^{\absn}$,
  which means that $\struct{M}$ is not $\mname{L{-}NDI}$-secure for $\pol$.
  \qed
\end{proof}

For the last relationship, the notion High-up $\mname{NDI}$ is compared with $\mname{NDI}_{pw}$.
Recall that $\mname{NDI}_{pw}$ doesn't properly deal with combined behaviour,
but this is what one would expect from a sensible Nondeducibility notion for multi-domain policies.
The statement made by the next proposition therefore shouldn't be interpreted as redeeming.

\begin{proposition}
  $\mname{H{-}NDI}$ implies $\mname{NDI}_{pw}$.
\end{proposition}
\begin{proof}
  Let $\struct{M}$ be a system with domain set $D$ and $\pol$ a policy over $D$.
  Assume that $\struct{M}$ is not $\mname{NDI}_{pw}$-secure for $\pol$,
  then there are $u,v \in D$ with $u \not\pol v$ and $u \flows_I v$.
  We can proceed as in the proof of Theorem \ref{thm:LNDI-implies-NDI}:
  the system $\struct{M}^{\cal D}$ has a vulnerability $(\set{v},\alpha,\beta,\pol^{\cal D})$,
  where ${\cal D} := \pset{ \set{w} } { w \in D }$,
  and choosing $\absn := \absUp{u}{}$ yields a vulnerability of $\struct{M}^{\absn}$ via Lemma \ref{thm:NDI-vuln-lemma}.
  
  (An important point here is that $\alpha$ consists of actions by a single domain only,
  whereas in the proof of Theorem \ref{thm:LNDI-implies-NDI} the sequence $\alpha$ can contain actions by multiple domains.
  If only a single domain $u$ is acting, a High-up cut can capture $u$ in its abstracted High domain;
  in the case of multiple active domains it might not, see Example \ref{ex:NDI-fruitfly}.)
  \qed
\end{proof}


\subsection{Relationships between the GN and NDI Variants}


\begin{proposition}
  The notion $\mname{GN}_{pw}$ doesn't imply any of $\mname{H{-}NDI}$, $\mname{L{-}NDI}$ or $\mname{C{-}NDI}$.
  The notion $\mname{L{-}GN}$ does not imply $\mname{H{-}NDI}$ or $\mname{C{-}NDI}$.
\end{proposition}
\begin{proof}
  The system depicted in Figure \ref{fig:GN-fruitfly}
  is $\mname{GN}_{pw}$- but not $\mname{H{-}GN}$-secure as argued in the proof of Theorem \ref{thm:downcuts-dont-cut-it}
  for the corresponding policy
  and therefore not $\mname{H{-}NDI}$-secure (since $\mname{H{-}GN}$ implies $\mname{H{-}NDI}$).
  As a trivial consequence, we get that $\mname{GN}_{pw}$ doesn't imply $\mname{C{-}NDI}$.
  Also, we have that $\mname{L{-}GN}$ doesn't imply $\mname{H{-}NDI}$,
  since otherwise $\mname{H{-}NDI}$ implied $\mname{L{-}GN}$ and, since $\mname{L{-}GN}$ implies $\mname{L{-}NDI}$,
  also $\mname{L{-}NDI}$, contradicting Theorem \ref{thm:LNDI-incomparable-HNDI}.
  As a consequence, we find that $\mname{L{-}GN}$ doesn't imply $\mname{C{-}NDI}$.

  To see that $\mname{GN}_{pw}$ doesn't imply $\mname{L{-}NDI}$,
  add the edge $(L_1,L_2)$ to the policy considered in the previous paragraph,
  and consider the cut $\absDown{L_2}{}$, which is given by $(\set{H}, \set{L_1,L_2})$
  and is the same cut used in the proof of Theorem \ref{thm:downcuts-dont-cut-it} to demonstrate a violation of $\mname{H{-}NDI}$-security.
  Therefore, the system is not $\mname{L{-}NDI}$-secure either.

  \qed
\end{proof}

If all results from this subsection are combined, we obtain exactly the containment diagram as claimed by Theorem \ref{thm:main}.


\subsection{Monotonicity}

The statement $u \not\pol v$ can be understood as a noninterference constraint
and adding the edge $u \pol v$ removes this constraint from a policy.
If a system is secure (for a sensible definition of `secure') and constraints are discarded from the policy,
it seems reasonable to expect that security is preserved.
In this subsection we investigate which of our notions support this intuition.

We have to compare cuts of the same domain set but with respect to different policies,
which is why we make explicit which policy a cut refers to by writing, for example, $\absUpPol{\cdot}{\pol}$.

If not mentioned otherwise, all systems in this subsection refer to their set of domains as $D$,
and we have two policies $\pol_0$ and $\pol_1$ with $\pol_0 \subseteq \pol_1$.

\begin{theorem}
  \label{thm:GN-monotonic}
  The notions $\mname{GN}$, $\mname{H{-}GN}$ and $\mname{C{-}GN}$ are monotonic.
\end{theorem}
\begin{proof}
  \emph{$\mname{GN}$ is monotonic:}
  Let $(u, \alpha_0, a, \alpha_1, \beta, \pol_1)$ be a $\mname{GN}$-vulnerability of some system $\struct{M}$,
  then $\dom(a) \not\pol_1 u$,
  which implies $\dom(a) \not\pol_0 u$ since $\not\pol_1 \subseteq \not\pol_0$.
  Because domain assignments and $u$ views are not affected by the policy,
  we have that $(u, \alpha_0, a, \alpha_1 ,\beta, \pol_0)$ is a $\mname{GN}$ vulnerability of $\struct{M}$.

  \emph{$\mname{H{-}GN}$ and $\mname{C{-}GN}$ are monotonic:}
  Since  $\mname{H{-}GN}$ and $\mname{C{-}GN}$ are equivalent by Theorem \ref{thm:CGN=HGN},
  it suffices to prove it for $\mname{H{-}GN}$ only.
  
  Set $\absn_1 := \absUpPol{\dom(a)}{\pol_1}$,
  ${\cal H}_1 := \dom(a)^{\pol_1}$ and ${\cal L}_1 := D \setminus \dom(a)^{\pol_1}$.
  Let $\struct{M}$ be a system such that $({\cal L}_1,\alpha_0,a,\alpha_1,\beta,\pol_1)$ is a
  $\mname{GN}$ vulnerability of $\struct{M}^{\absn}$.
  Then we have ${\cal H}_1 \not\pol_1^{\absn_1} {\cal L}_1$.
  Furthermore, set
  $\absn_0 := \absUpPol{\dom(a)}{\pol_0}$,
  ${\cal H}_0 := \dom(a)^{\pol_0}$, and ${\cal L}_0 := D \setminus \dom(a)^{\pol_0}$.
  This implies ${\cal H}_0 \not\pol_0^{\absn_0} {\cal L}_0$.
  We show that $({\cal L}_0, \alpha_0, a, \alpha_1, \pol_0)$ is a $\mname{GN}$ vulnerability of $\struct{M}^{\absn_0}$ by
  demonstrating that we have ${\cal L}_1 \subseteq {\cal L}_0$ and $\dom^{\absn_0}(a) \not\pol_0^{\absn_0} {\cal L}_0$,
  and then applying Lemma \ref{thm:GN-vuln-lemma}.

  For the former, let $u \in {\cal L}_1$.
  Then $\dom(a) \not\pol_1 u$ which implies $\dom(a) \not\pol_0 u$,
  and this gives us $u \in {\cal L}_0$.
  For the latter, by choice of $\absn_0$ we clearly have that $\dom(a) \in {\cal H}_0$ and therefore $\dom^{\absn_0}(a) = {\cal H}_0$,
  which implies $\dom^{\absn_0}(a) \not\pol_0^{\absn_0} {\cal L}_0$.
  Lemma \ref{thm:GN-vuln-lemma} now gives us the existence of a $\mname{GN}$ vulnerability of $\struct{M}^{\absn_0}$.
  \qed
\end{proof}
That High-up $\mname{GN}$ is monotonic can be explained with the fact that adding edges can never grow the Low
component in a given High-up cut.
Since according to the definition, High is taken to be $u^{\pol}$ for a given domain $u$, adding an edge
might grow High, which in turn would shrink Low.
So adding edges can never increase the knowledge of the Low coalition when High-up cuts are used.

The case is different for Low-down $\mname{GN}$.
Adding edges to a policy can join two formerly incomparable elements
which then become members of Low in a Low-down coalition.
The reason is the definition of the Low component, which for a given domain $u$ is taken to be ${}^{\pol}u$,
or in other words, Low will be all domains from which $u$ is permitted to learn.
Protection due to Low-down cuts thus requires a somewhat friendly attacker,
so there is an argument that this can be considered a flaw in the definition itself.

\begin{proposition}
  $\mname{L{-}GN}$ is not monotonic.
\end{proposition}
\begin{proof}
  As shown in the proof of Theorem \ref{thm:downcuts-dont-cut-it},
  the system depicted in Figure \ref{fig:GN-fruitfly} is $\mname{L{-}GN}$-secure.
  The proof of Theorem \ref{thm:GN-neq-downGN} demonstrates that adding the edge $(L_1,L_2)$ to the policy turns it non-$\mname{L{-}GN}$-secure.
  \qed
\end{proof}


This concludes the investigation of monotonicity for our variants of $\mname{GN}$.
But there are two more things to note about adding edges to a policy when $\mname{GN}$ is used:
(1) The attack surface does not increase, since already generic
$\mname{GN}$ requires all domains to be unable to infer the occurrence or nonoccurrence of \emph{any} secret events.
In fact, it might even get smaller.
(2) Adding edges can decrease the number of cuts, so there might be fewer requirements on the system in order to
be secure.
For Cut $\mname{GN}$, the results from Theorem \ref{thm:GN-monotonic} suggests that a smaller number of cuts can compensate for
a possible increase of Low's knowledge.

The argument for Cut $\mname{NDI}$ being monotonic is similar to the case of High-up and Cut $\mname{GN}$.
That $\mname{NDI}$ and $\mname{NDI}_{pw}$ are monotonic is straightforward.

\begin{proposition}
  The notions $\mname{NDI}_{pw}$, $\mname{NDI}$ and $\mname{C{-}NDI}$ are monotonic.
\end{proposition}
\begin{proof}
  \emph{$\mname{NDI}_{pw}$ is monotonic:}
  Let $\struct{M}$ be $\mname{NDI}_{pw}$-secure for $\pol_0$ and $u, v \in D$ such that $u \flows_I v$.
  Then $u \pol_0 v$ by $\mname{NDI}_{pw}$-security of $\struct{M}$ for $\pol_0$ and $u \pol_1 v$ by $\pol_0 \subseteq \pol_1$.
  Therefore, $\struct{M}$ is $\mname{NDI}_{pw}$-secure for $\pol_1$.
  

  \emph{$\mname{C{-}NDI}$ is monotonic:}
  Let $\struct{M}$ be a system and $\absn = ({\cal H}, {\cal L})$ be a cut of $D$, then we have $H \not\pol_1^\absn L$.
  Assume $\struct{M}^\absn$ has an $\mname{NDI}$ vulnerability $({\cal L},\alpha,\beta,\pol_1)$.
  We show that $({\cal L},\alpha,\beta,\pol_0)$ is a vulnerability of that system, too.

  First, let $u \in {\cal H}$ and $v \in {\cal L}$.
  Then we have $u \not\pol_1 v$, which implies $u \not\pol_0 v$.
  Therefore ${\cal H} \not\pol_0^{\absn} {\cal L}$ and $\absn$ is a valid cut with respect to $\pol_0$.
  This also gives us $\dom^\absn(a) \not\pol_0^\absn {\cal L}$ for all actions $a$ that occur in $\alpha$,
  because we must have $\dom^\absn(a) = {\cal H}$ by vulnerability.
  
  Finally, clearly $\beta$ remains a valid ${\cal L}$ view and trivially we have ${\cal L} \subseteq {\cal L}$.
  Therefore, due to Lemma \ref{thm:NDI-vuln-lemma}, we get that $({\cal L},\alpha,\beta,\pol_0)$ is an $\mname{NDI}$
  vulnerability of $\struct{M}^\absn$.

  \emph{$\mname{NDI}$ is monotonic:}
  Let $\struct{M}$ be a system that is $\mname{NDI}$-secure for $\pol_0$,
  $u \in D$,
  $X_0 := {}^{\not\pol_0}u$, $X_1 :={}^{\not\pol_1}u$,
  $\alpha \in {A_{X_1}}^*$ and $\beta \in \Views_u(\struct{M})$.
  Then, because of $\not\pol_1 \subseteq \not\pol_0$, we have $X_1 \subseteq X_0$,
  which implies $\alpha \in {A_{X_0}}^*$.
  With $\mname{NDI}$-security of $\struct{M}$ for $\pol_0$ we get that $\alpha$ and $\beta$ are compatible,
  and therefore $\struct{M}$ is $\mname{NDI}$-secure for $\pol_1$ as well.
  \qed
\end{proof}

Merging two incomparable domains into a Low-down cut is also possible for $\mname{NDI}$,
and here we too obtain the result that the Low-down variant fails to be monotonic.
In fact, the same system as in the case of Low-down $\mname{GN}$ can be used.

However, the notion High-up $\mname{NDI}$ does not share the monotonicity property with its $\mname{GN}$ counterpart.
While it's true that Low might shrink if edges are added and thus has less knowledge at hand for an attack,
and that the new policy might have fewer cuts,
the set of action sequences that have to be compatible with any Low view grows,
which increases the attack surface (i.e., Low might now be able to exclude certain High behaviours).
The next result suggests that this increase can outweigh the loss of knowledge experienced by Low and
the fewer number of cuts combined.

\begin{theorem}
  The notions $\mname{L{-}NDI}$ and $\mname{H{-}NDI}$ are not monotonic.
\end{theorem}
\begin{proof}
  \emph{$\mname{L{-}NDI}$ is not monotonic:}
  In the proof of Theorem \ref{thm:LNDI-incomparable-HNDI} it is shown that the the system in Figure \ref{fig:GN-fruitfly} is $\mname{L{-}NDI}$-secure for the depicted policy, which we call $\pol_0$,
  given by $L_1 \pol_0 H$ and $L_2 \pol_0 H$ (excluding edges due to reflexivity).
  Let $\pol_1$ be the policy obtained by taking $\pol_0$ and adding the additional edge $L_1 \pol_1 L_2$.
  Then the system is not $\mname{L{-}NDI}$-secure for $\pol_1$:
  consider the cut $\absDownPol{L_2}{\pol_1}$.
  It is equivalent to $\absUpPol{H}{\pol_0}$.
  In the proof of Theorem \ref{thm:LNDI-incomparable-HNDI},
  it is demonstrated that the system is not $\mname{NDI}$-secure for $\pol_0^{\absUp{H}{}}$.
  But since $\absUp{H}{}$ and $\absDown{L_2}{}$ are equal,
  we get that it isn't $\mname{NDI}$-secure for $\pol_1^{\absDown{L_2}{}}$ either.
  
  \emph{$\mname{H{-}NDI}$ is not monotonic:}
  Consider the system $\struct{M}$ in Example \ref{ex:NDI-fruitfly} and call the depicted policy $\pol_0$,
  which is given by $L \pol_0 H_1$ and $L \pol_0 H_2$, excluding edges due to reflexivity.
  Let $\pol_1$ be the policy $\pol_0$ with the additional edge $H_1 \pol_1 H_2$.
  The system is $\mname{H{-}NDI}$-secure for $\pol_0$ due to Theorem \ref{thm:LNDI-incomparable-HNDI},
  but it is not $\mname{H{-}NDI}$-secure for $\pol_1$.
  To see this, take the cut $\absUpPol{H_1}{\pol_1}$,
  which is equivalent to $\absDownPol{L}{\pol_0}$.
  And as argued in the proof of Proposition \ref{thm:NDI-strictly-contains-NDIN},
  the system $\struct{M}^{\absDownPol{L}{\pol_0}}$ has an $\mname{NDI}$ vulnerability,
  and thus $\struct{M}$ is not $\mname{H{-}NDI}$-secure for $\pol_1$.
  \qed
\end{proof}

Since $\mname{NDI}$ is monotonic, it cannot be equal to Low-down $\mname{NDI}$, and
therefore this containment must be strict.
With the same argument we get strict containment of High-up $\mname{NDI}$ in pointwise $\mname{NDI}.$

\begin{corollary}
  $\mname{L{-}NDI}$ is strictly contained in $\mname{NDI}$, and
  $\mname{H{-}NDI}$ is strictly contained in $\mname{NDI}_{pw}$.
\end{corollary}

Together with the results on containment relationships obtained in the previous subsections,
our results on monotonicity now yield proofs for Theorem \ref{thm:main}.

} 
{
}

\section{Conclusion}
\label{sec:concl} 

In this work we have discussed several variants of Generalized Noninterference and Nondeducibility on Inputs
for multi-domain policies that use reductions to the two-level case,
including a technique proposed by Ryan.
We have found that this technique leads to a stricter notion in the case of Generalized Noninterference,
but behaves counter-intuitively in the case of Nondeducibility on Inputs,
where it yields a notion that is incomparable to a natural variant for multi-domain policies.
We have found evidence that seems to suggest that considering all cuts is a more robust choice as a reduction technique.
Some notions we obtained break our intuitions in the sense that they are not preserved under removing noninterference constraints.

These results have left open a question about how to handle the general case of collusion, as
reductions to $H \not\pol L$ are a special case of collusion where two coalitions are operating,
while general abstractions can model an arbitrary number of coalitions.
It seems natural to extend the theory such that it can handle general abstractions,
but then we leave the area of transitive noninterference.
For example, consider the transitive policy $\pol$ that contains the relations $A \pol B$ and $C \pol D$ only, and
the abstraction ${\cal D}$ that forms the coalitions $\set{A}$, $\set{B,C}$ and $\set{D}$.
The resulting policy $\pol^{\cal D}$ is intransitive, as it 
has edges $\set{A} \pol^{\cal D} \set{B,C}$ and $\set{B,C} \pol^{\cal D}\set{D}$, but 
lacks the edge $\set{A} \pol^{\cal D} \set{D}$.
In this case, it seems reasonable to say that information may get from $A$ to $D$,
as domains $B$ and $C$ collude and share their observations, but it needs intermediate behaviour by them
in order to forward the information.
Adding the edge $\set{A} \pol^{\cal D} \set{D}$ clashes with this reasoning,
as it would express that $A$ may \emph{directly} communicate with $D$.
This suggests that dealing with general abstractions requires techniques from the theory of intransitive noninterference.
Semantics for intransitive noninterference that build in types of collusion 
have been considered in a few works \cite{BackesP03,RonZhangEngelhardt2012},
but the relationship of these definitions to abstractions remains to be studied.

\bibliographystyle{splncs}
\bibliography{main,fastni}

\begin{thebibliography}{10}

\bibitem{HaighYoung87}
Haigh, J.T., Young, W.D.:
\newblock {Extending the noninterference version of MLS for SAT}.
\newblock IEEE Transactions on Software Engineering \textbf{13}(2) (1987)  141

\bibitem{rushby_92}
Rushby, J.:
\newblock Noninterference, transitivity, and channel-control security policies.
\newblock Technical report, SRI international (Dec 1992)

\bibitem{RonEsorics2007}
van~der Meyden, R.:
\newblock What, indeed, is intransitive noninterference?
\newblock Journal of Computer Security \textbf{23}(2) (2015)  197--228
  (Extended version of a paper in ESORICS'2007.).

\bibitem{GM82}
Goguen, J.A., Meseguer, J.:
\newblock {Security Policies and Security Models}.
\newblock In: 1982 {IEEE} Symposium on Security and Privacy, Oakland, CA, USA,
  April 26-28, 1982. (1982)  11--20

\bibitem{Sutherland86}
Sutherland, D.:
\newblock {A model of information}.
\newblock In: Proc. 9th National Computer Security Conference, DTIC Document
  (1986)  175--183

\bibitem{McCullough88Ulysses}
McCullough, D.:
\newblock {Foundations of Ulysses: The theory of security}.
\newblock Technical report, DTIC Document (1988)

\bibitem{McLean94}
McLean, J.:
\newblock {A general theory of composition for trace sets closed under
  selective interleaving functions}.
\newblock In: Research in Security and Privacy, 1994. Proceedings., 1994 IEEE
  Computer Society Symposium on, IEEE (1994)  79--93

\bibitem{Mantel2000Possibilistic}
Mantel, H.:
\newblock {Possibilistic definitions of security - an assembly kit}.
\newblock In: Computer Security Foundations Workshop, 2000. CSFW-13.
  Proceedings. 13th IEEE, IEEE (2000)  185--199

\bibitem{focardi_01}
Focardi, R., Gorrieri, R.:
\newblock Classification of security properties.
\newblock In: {FOSAD} 2000, {LNCS} 2171. (2001)  331--396

\bibitem{roscoe_95}
Roscoe, A.W.:
\newblock {CSP} and determinism in security modelling.
\newblock In: Proc. IEEE Symposium on Security and Privacy. (1995)  114--221

\bibitem{Ryan2000}
Ryan, P.Y.:
\newblock {Mathematical models of computer security}.
\newblock In: Foundations of Security Analysis and Design.
\newblock Springer (2000)  1--62

\bibitem{forster_97}
Forster, R.:
\newblock Non-interference properties for nondeterministic processes.
\newblock PhD thesis, Dissertation for transfer to D.Phil status, Oxford
  University Computing Laboratory (1997)

\bibitem{Mantel_thesis}
Mantel, H.:
\newblock A uniform framework for the formal specification and verification of
  information flow security.
\newblock PhD thesis, Universit\"{a}t des Saarlandes (2003)

\bibitem{millen_94}
Millen, J.K.:
\newblock Unwinding forward correctability.
\newblock In: Proc. IEEE Computer Security Foundations Workshop. (1994)  2--10

\bibitem{RoscoeWW96}
Roscoe, A.W., Woodcock, J., Wulf, L.:
\newblock Non-interference through determinism.
\newblock Journal of Computer Security \textbf{4}(1) (1996)  27--54

\bibitem{sutherland_86}
Sutherland, D.:
\newblock A model of information.
\newblock In: Proc. National Computer Security Conference. (1986)  175--183

\bibitem{McCullough88}
McCullough, D.:
\newblock {Noninterference and the composability of security properties}.
\newblock In: Proceedings of the 1988 {IEEE} Symposium on Security and Privacy,
  Oakland, California, USA, April 18-21, 1988. (1988)  177--186

\bibitem{RonSebastian2016}
{Eggert}, S., {van der Meyden}, R.:
\newblock {Dynamic Intransitive Noninterference Revisited}.
\newblock CoRR (2016) arXiv:1601.05187 [cs.CR].

\bibitem{GM84}
Goguen, J.A., Meseguer, J.:
\newblock {Unwinding and Inference Control}.
\newblock In: Proceedings of the 1984 {IEEE} Symposium on Security and Privacy,
  Oakland, California, USA, April 29 - May 2, 1984. (1984)  75--87

\bibitem{RonChenyi2007}
van~der Meyden, R., Zhang, C.:
\newblock {Algorithmic Verification of Noninterference Properties}.
\newblock Electr. Notes Theor. Comput. Sci. \textbf{168} (2007)  61--75

\bibitem{RonZhangEngelhardt2012}
Engelhardt, K., van~der Meyden, R., Zhang, C.:
\newblock {Intransitive Noninterference in Nondeterministic Systems}.
\newblock In: Proceedings of the 2012 ACM conference on Computer and
  communications security, ACM (2012)  869--880

\bibitem{BackesP03}
Backes, M., Pfitzmann, B.:
\newblock Intransitive non-interference for cryptographic purposes.
\newblock In: {IEEE} Symposium on Security and Privacy. (2003)  140--152

\end{thebibliography}

\end{document}